\documentclass[a4paper,UKenglish,cleveref,autoref,thm-restate]{lipics-v2021}

\title{Fast Shortest Path in Graphs With Sparse Signed Tree Models and Applications}
\titlerunning{Fast Shortest Path in Graphs With Sparse Signed Tree Models and Applications}

\author{\'{E}douard Bonnet}{CNRS, ENS de Lyon, Université Claude Bernard Lyon 1, LIP UMR 5668, 69342 Lyon, France \and \url{http://perso.ens-lyon.fr/edouard.bonnet}}{edouard.bonnet@ens-lyon.fr}{https://orcid.org/0000-0002-1653-5822}{}
\author{Colin Geniet}{Discrete Mathematics Group, Institute for Basic Science (IBS), Daejeon, South Korea \and \url{https://colingeniet.com/}}{research@colingeniet.com}{https://orcid.org/0000-0003-4034-7634}{}
\author{Eun Jung Kim}{School of Computing, KAIST, Daejeon, South Korea \and \url{https://ssimplexity.github.io/}}{eunjungkim78@gmail.com}{https://orcid.org/0000-0002-6824-0516}{}
\author{Sungmin Moon}{School of Computing, KAIST, Daejeon, South Korea}{msm241_23@kaist.ac.kr}{}{}

\authorrunning{\'E. Bonnet, C. Geniet, E. J. Kim, S. Moon}

\Copyright{Édouard Bonnet, Colin Geniet, Eun Jung Kim, Sungmin Moon}

\category{}

\relatedversion{}

\supplement{}

\funding{}

\acknowledgements{We thank Hung Le for clarifying answers regarding \cite[Lemma 7.4]{Chan25}.}

\nolinenumbers 

\hideLIPIcs  

\EventEditors{John Q. Open and Joan R. Access}
\EventNoEds{2}
\EventLongTitle{42nd Conference on Very Important Topics (CVIT 2016)}
\EventShortTitle{CVIT 2016}
\EventAcronym{CVIT}
\EventYear{2016}
\EventDate{December 24--27, 2016}
\EventLocation{Little Whinging, United Kingdom}
\EventLogo{}
\SeriesVolume{42}
\ArticleNo{23}

\usepackage[utf8]{inputenc}  

\usepackage[T1]{fontenc}
\usepackage[colorinlistoftodos,bordercolor=orange,backgroundcolor=orange!20,linecolor=orange,textsize=normalsize]{todonotes}
\usepackage{amsmath}
\usepackage{amssymb}
\usepackage{bbm}
\usepackage{accents}
\usepackage{complexity}
\usepackage{booktabs}
\usepackage{bm}
\usepackage{fixmath}
\usepackage{tcolorbox}

\definecolor{amber}{rgb}{1.0, 0.75, 0.0}

\usepackage{xspace}
\usepackage{tikz}
\usepackage[ruled,vlined,linesnumbered]{algorithm2e}

\usetikzlibrary{fit}
\usetikzlibrary{arrows,arrows.meta}
\usetikzlibrary{patterns}
\usetikzlibrary{calc}
\usetikzlibrary{shapes}
\usetikzlibrary{positioning}
\usetikzlibrary{math}
\usetikzlibrary{shapes.symbols}
\usetikzlibrary{decorations.pathreplacing,calligraphy}
\usetikzlibrary{decorations.pathmorphing}
\usepackage[scr=boondox,scrscaled=1.05]{mathalfa}
\usetikzlibrary{shapes.geometric}

\renewcommand{\geq}{\geqslant}
\renewcommand{\leq}{\leqslant}
\renewcommand{\preceq}{\preccurlyeq}

\renewcommand{\le}{\leq}
\renewcommand{\ge}{\geq}

\newcommand{\defproblem}[3]{
  \vspace{1mm}
  \begin{tcolorbox}[
    colframe=black,        
    colback=white,         
    boxrule=0.5pt,         
    arc=4pt,               
    left=6pt, right=6pt,   
    top=6pt, bottom=6pt    
  ]    
    #1 \\
    {\bf{Input:}} #2 \\
    {\bf{Output:}} #3
  \end{tcolorbox}
  \vspace{1mm}
}

\newtheorem{question}{Question}

\crefname{claim}{Claim}{Claims}

\newcommand\abs[1]{\lvert #1\rvert}

\newcommand\sd{\text{sd}}

\newcommand\adj{\text{Adj}}

\newcommand{\ceil}[1]{\lceil #1 \rceil}

\newcommand{\Nn}{\mathbb{N}}

\newcommand{\Bc}{\mathcal{B}}

\newcommand{\Fc}{\mathcal{F}}

\newcommand{\Pc}{\mathcal{P}}
\newcommand{\Rc}{\mathcal{R}}

\newcommand{\Tc}{\mathcal{T}}

\newcommand{\dist}{\mathsf{dist}}
\newcommand{\Ball}{\mathsf{Ball}}

\newcommand{\mw}{\mathsf{mw}}

\newcommand{\qgt}{\succ}

\begin{document}

\maketitle

\begin{abstract}
  A~signed tree model of a~graph $G$ is a~compact binary structure consisting of a~rooted binary tree whose leaves are bijectively mapped to the vertices of $G$, together with 2-colored edges $xy$, called transversal pairs, interpreted as bicliques or anti-bicliques whose sides are the leaves of the subtrees rooted at $x$ and at $y$. 
  We design an algorithm that, given such a~representation of an unweighted $n$-vertex graph~$G$ with $p$ transversal pairs, and given a~source $v \in V(G)$, computes a~shortest-path tree rooted at $v$ in $G$ in time $O(p \log n)$.
  A~wide variety of graph classes are such that for all $n$, their $n$-vertex graphs admit signed tree models with $O(n)$ transversal pairs: for instance, those of bounded symmetric difference (hence, in particular, those of bounded flip-width, merge-width, twin-width, and degeneracy), more generally of bounded sd-degeneracy, as well as interval graphs.
  
  As applications of our \textsc{Single-Source Shortest Path} algorithm and new techniques, we 
  \begin{itemize}
  \item improve the runtime of the fixed-parameter algorithm for first-order model checking on graphs given with a~witness of low merge-width from cubic [Dreier \& Toruńczyk, STOC '25] to quadratic;
  \item give an $O(n^2 \log n)$-time algorithm for \textsc{All-Pairs Shortest Path} on graphs given with a~witness of low merge-width, generalizing a~result known for twin-width [Twin-Width III, SICOMP '24];  
  \item significantly extend and simplify an $O(n^2 \log n)$-time algorithm for multiplying two $n \times n$ matrices $A, B$ of bounded twin-width in [Twin-Width V, STACS '23]: now $A$ solely has to be \emph{an} adjacency matrix of a~graph of bounded twin-width and $B$ can be arbitrary;
  \item give an $O(n^2 \log^2 n)$-time algorithm for \textsc{All-Pairs Shortest Path} on graphs of bounded twin-width, bypassing the need for contraction sequences in [Twin-Width III, SICOMP '24; Bannach et al. STACS '24];
  \item give an $O(n^{7/3} \log^2 n)$-time algorithm for \textsc{All-Pairs Shortest Path} on graphs of symmetric difference $O(n^{1/3})$.
  \end{itemize}
  The second and the last two items imply the same for \textsc{Diameter}, \textsc{Radius}, \textsc{Eccentricity}, \textsc{Wiener Index}, etc.
  The last three items do \emph{not} assume any witness to be given as part of the input.  
\end{abstract}

\section{Introduction}\label{sec:intro}

Shortly after twin-width was introduced~\cite{twin-width1}, it was observed that graphs $G$ of bounded twin-width admit natural sparse representations called \emph{twin-decompositions}~\cite{twin-width3} or \emph{tree models}~\cite{BonnetNMST24}: a~rooted binary tree $T$ whose $n$ leaves are in one-to-one correspondence with the vertices of~$G$, together with relatively few extra edges called \emph{transversal edges} such that two ``leaves'' $\ell$ and $\ell'$ are adjacent in~$G$ whenever an ancestor of $\ell$ and an ancestor of $\ell'$ are linked by a~transversal edge, and in which case, this pair of adjacent ancestors is unique.
Every graph admits a~tree model if no constraint is imposed on the number of transversal edges or on the graph that they induce.
Not every graph class, though, admits \emph{sparse} tree models, i.e., ones with $O(n)$ transversal edges, since this property implies that the class has factorial growth.
Graphs of bounded twin-width happen to even admit \emph{degenerate} tree models, i.e., where the graphs formed by the transversal edges (together with $T$) are of bounded degeneracy.

Degenerate tree models are generally useful to the theory of twin-width, and have yielded faster shortest-path algorithms when such representations are given as part of the input: an~$O(n \log n)$-time algorithm for \textsc{Single-Source Shortest Paths}~\cite{twin-width3}, improved to~$O(n)$ time~\cite{bannach2024dag}.
They are at play in characterizing bounded twin-width classes by first-order transductions of pattern-avoiding permutation classes~\cite{BonnetNMST24}, and an $O(n)$-time algorithm multiplying $n \times n$ matrices of bounded twin-width over a~finite ring~\cite{twin-width5}.
They are also implicitly leveraged in the \mbox{$\chi$-boundedness} of graphs of bounded twin-width~\cite{twin-width3}, which can be thought of as a~top-down coloring of the tree model, and in their $O(\log n)$-bit adjacency labeling schemes~\cite{twin-width2} after arranging that the tree has logarithmic height.  
(On the other hand, if the main concern is data compression, better data structures are possible; see~\cite{PilipczukSZ22}.)
For some applications, like fast shortest-path algorithms, sparse tree models are sufficient. 

At this point, we should note that degenerate tree models can encode graphs of unbounded twin-width.
For instance, graphs of bounded degeneracy trivially admit degenerate tree models by placing the transversal edges directly at the leaves of~$T$.
Aiming to design short adjacency labeling schemes for broad graph classes, the authors of~\cite{bonnet2024signedtree} further generalize tree models with \emph{signed tree models}.
These add transversal anti-edges $A$, which together with the transversal edges $B$ form the transversal pairs (see also~\cite{twin-width5} where the transversal edges are labeled).
Now two ``leaves'' $\ell$ and $\ell'$ are adjacent in~$G$ if and only if the ``lowest''---when representing the tree $T$ with the root up and the leaves down---transversal pair with one endpoint ancestor of~$\ell$ and the other endpoint ancestor of~$\ell'$ is an edge (and not an anti-edge).
A~non-crossing condition on the transversal pairs makes this well-defined.

Surprisingly many factorial graph classes admit sparse signed tree models, and even degenerate ones.
The \emph{sd-degeneracy} graph parameter is defined in~\cite{bonnet2024signedtree}, extending degeneracy to the dense world, and doing so more permissively than \emph{symmetric difference}~\cite{AtminasCLZ15}.
As graphs of bounded sd-degeneracy admit degenerate signed tree models~\cite{bonnet2024signedtree}, this is in particular true for the (less general) classes of bounded degeneracy, symmetric difference, twin-width, flip-width~\cite{Torunczyk23}, and merge-width~\cite{merge-width}.\footnote{Indeed, bounded twin-width $\Rightarrow$ bounded merge-width $\Rightarrow$ bounded flip-width $\Rightarrow$ bounded symmetric difference $\Rightarrow$ bounded sd-degeneracy, and bounded degeneracy $\Rightarrow$ bounded symmetric difference.}
It is relatively easy to see that interval graphs too admit degenerate signed tree models (for which the trees are combs).
Hereditary classes of linear neighborhood complexity\footnote{A~class $\mathcal C$ has \emph{linear neighborhood complexity} if for every $G \in \mathcal C$ and for every subset $X \subseteq V(G)$, the number of distinct neighborhoods within $X$ of vertices of~$G$ is~$O(|X|)$.} have bounded symmetric difference~\cite{Davies25}, hence admit degenerate signed tree models.
We will see that classes with almost\footnote{A~class $\mathcal C$ has \emph{almost linear neighborhood complexity} if for every $G \in \mathcal C$ and for every subset $X \subseteq V(G)$, the number of distinct neighborhoods within $X$ of vertices of~$G$ is~$|X|^{1+o(1)}$.} linear neighborhood complexity admit signed tree models with~$n^{1+o(1)}$ transversal pairs; we may call them \emph{almost sparse signed tree models}.
It is conjectured that every monadically dependent class has almost linear neighborhood complexity~\cite{Dreier24}.

The generality of classes with (almost) sparse signed tree models motivates devising fast algorithms on these representations.
We first show how to efficiently convert signed tree models into other sparse encodings.
In the next theorem, an \emph{interval biclique partition} of a~graph $G$ is a~set of bicliques edge-partitioning $G$ together with a~linear order on the vertices of~$G$ such that each side of every biclique is an interval along this order.
We defer the definition of a~\emph{DAG compression}~\cite{bannach2024dag} to the next section.
The reader can now think of it as a~sparse digraph encoding adjacencies via reachability.  

\begin{restatable}{theorem}{mainthm}\label{thm:main}
  There is an $O(p \log n)$-time algorithm that converts a~signed tree model with $p$ transversal pairs of an $n$-vertex graph $G$ into each of the following:
  \begin{enumerate}
  \item an interval biclique partition of $G$ with $O(p)$ bicliques,
  \item a~DAG compression of size $O(p \log n)$ for $G$, or
  \item with additional $O(p \log^2 n)$ time, a~(positive) tree model of~$G$ with $O(p \log^2 n)$ transversal edges.
  \end{enumerate}
\end{restatable}

Either of the first two items (the first combined with~\cite{twin-width3}, the second, with~\cite{bannach2024dag}) implies an equally fast \textsc{Single-Source Shortest Path} algorithm for sparsely encoded graphs.

\begin{theorem}\label{thm:main-sssp}
  There is an $O(p \log n)$-time algorithm that, given a~signed tree model with $p$ transversal pairs of an $n$-vertex graph $G$ and $v \in V(G)$, outputs a~shortest-path tree of~$G$ rooted at~$v$.
\end{theorem}

Notably, for sparse tree models (i.e., when $p=O(n)$), \cref{thm:main-sssp} solves \textsc{Single-Source Shortest Path} in time $O(n \log n)$; and in time $n^{1+o(1)}$ for almost sparse tree models.
This is essentially optimal, and for graphs with $\Omega(n^{1+\varepsilon})$ edges for some $\varepsilon > 0$, these running times are sublinear in the size of~$G$ itself.

The algorithm behind \cref{thm:main-sssp} uses routines from computational geometry.
It works as follows.
We first compute a~forest $F$ on top of the transversal pairs of the signed tree model $\mathcal T=(T,A,B)$, where $e$ is the~parent of~$f$ whenever $e$ is \emph{above} $f$ (i.e., each endpoint of $e$ is an ancestor of an endpoint of~$f$) and there is no transversal pair $e'$ that is \emph{in between} (i.e., such that $e$ is above $e'$, and $e'$ is above $f$).
We can clean $F$ so that every edge in $E(F)$ is between transversal pairs of $\mathcal T$ of distinct signs (i.e., on distinct sets among $A,B$).
Now consider a~transversal edge $e \in V(F)$ and all its children $e_1, \ldots, e_h$, which are transversal anti-edges.
Note that $e$ corresponds to a~positive rectangle $R := X \times Y$ in~$G$ where $X$ and $Y$ are the vertices at the leaves of the subtree of $T$ rooted at each endpoint of~$e$.
Similarly, each $e_i$ induces a~negative rectangle $R_i \subset R$.
We thus wish to complement the pairwise-disjoint rectangles $R_1, \ldots, R_h$ within $R$ in the form of a~suitably bounded number of pairwise-disjoint rectangles.
We show that this can be done with $O(h)$ rectangles in $O(h \log h)$ time.
Applying this step in parallel for every transversal edge $e$, we obtain a~partition of~$G$ into $O(p)$ bicliques whose sides are intervals; hence, an interval biclique partition.
This suffices to prove~\cref{thm:main-sssp} invoking~\cite[Theorem 6.2]{twin-width3}.
The rest of~\cref{thm:main} is established by replacing~$T$ with a~balanced binary tree, so that each biclique can be realized with few transversal edges.

As one can solve \textsc{All-Pairs Shortest Path} with $n$ calls to \textsc{Single-Source Shortest Path},~\cref{thm:main-sssp} has the following consequence.

\begin{corollary}\label{cor:apsp-gen}
  Let $\mathcal C$ be a~class of unweighted graphs such that there is a~(randomized) algorithm that inputs an $n$-vertex $m$-edge graph of~$\mathcal C$ and outputs a~signed tree model of~$G$ with at~most $p(n)$ transversal pairs in time $O(T(n,m))$, for two functions $p$ and $T$.
  
  Then, there is a~(randomized) algorithm that solves \textsc{All-Pairs Shortest Path} in time $O(T(n,m)+n \cdot p(n) \log n)$ on~$n$-vertex graphs of~$\mathcal C$.
\end{corollary}
    
Using the functional equivalence of twin-width and versatile twin-width~\cite[Lemma 3.8]{twin-width2}, we show via~\cref{thm:vers-tww-to-sd,lem:sd-degen-to-stm} the following. 
\begin{theorem}\label{thm:vers-tww-to-stm}
  Let $\mathcal C$ be any class of bounded twin-width.
  There is an algorithm that inputs an $n$-vertex $m$-edge graph $G \in \mathcal C$, and outputs a~signed tree model of~$G$ with $O(n \log n)$ transversal pairs, in time $O((m+n) \log n)$ with high probability.\footnote{That is, with probability at~least~$1-n^{-c}$ for any fixed but arbitrarily large constant $c>0$.}
\end{theorem}

Thus, by~\cref{cor:apsp-gen,thm:vers-tww-to-stm}:
\begin{theorem}\label{thm:apsp-tww}
  Let $\mathcal C$ be any class of bounded twin-width.
  \textsc{All-Pairs Shortest Path} in~$\mathcal C$ has an algorithm in time $O(n^2 \log^2 n)$ with high probability.
\end{theorem}

This implies a~similar running time for \textsc{Diameter}, \textsc{Radius}, \textsc{Eccentricity}, \textsc{Wiener Index}, etc.
Previously, such algorithms (with fewer $\log n$ factors) were only known when a~contraction sequence (that is, witness of low twin-width) was provided as part of the input~\cite{twin-width3,bannach2024dag}.
\Cref{thm:apsp-tww} is close to optimal: under the Strong Exponential-Time Hypothesis, for any $\varepsilon > 0$, \textsc{Diameter} cannot be solved in time $O(n^{2-\varepsilon})$ on $n$-vertex graphs of twin-width at~most~4 (even with a~witness)~\cite{twin-width3,Berge22}.

On the more general classes of bounded symmetric difference (even when it is only bounded by $O(n^{1/3})$), we show in~\cref{thm:sym-diff-sd-degen} how to find signed tree models with $O(n^{4/3} \log n)$ transversal pairs in time $O(n^{7/3} \log n)$.
\Cref{cor:apsp-gen} then yields: 

\begin{theorem}\label{thm:sym-diff}
  There is an algorithm that solves \textsc{All-Pairs Shortest Path} on $n$-vertex graphs of symmetric difference~$O(n^{1/3})$, in time $O(n^{7/3} \log^2 n)$ with high probability.
\end{theorem}

Recall that the fastest \textsc{All-Pairs Shortest Path} for \emph{general} unweighted graphs takes $O(n^\omega \log n)$ time~\cite{Seidel95}, where $\omega$ is the matrix-multiplication time exponent, currently at $2.371339 > \frac{7}{3} = 2.333\dots$~\cite{Alman25}, and that no purely combinatorial truly subcubic algorithm is known for this problem.

\cref{thm:main,thm:apsp-tww} can be combined to significantly extend an $O(n^2 \log n)$-time algorithm multiplying two $n \times n$ matrices of bounded twin-width~\cite[Theorem 5]{twin-width5}.
\begin{restatable}{theorem}{matmult}\label{thm:mat-mult}
  Let $\mathcal C$ be any class of bounded twin-width and $\mathcal A$ be any additive group.
  There is an algorithm that, given any adjacency matrix $M$ of an $n$-vertex graph of~$\mathcal C$ and any $n \times n$ matrix $N$ over $\mathcal A$, computes $MN$ in time $O(n^2 \log n)$ with high probability. 
\end{restatable}
The proof of~\cref{thm:mat-mult} has three steps.
We first compute a~signed tree model of the graph with $O(n \log n)$ transversal pairs, by~\cref{thm:vers-tww-to-stm}. 
We then convert it to an interval biclique partition $(\mathcal B,\prec)$ with $|\mathcal B|=O(n \log n)$ bicliques, by~\cref{thm:main}.
Let $M'$ be the reordering of~$M$ following~$\prec$, thus $M = PM'P^T$ for some permutation matrix defined by the original order of~$M$ and $\prec$.
We give a~simple algorithm (\cref{thm:ibp-vm-mult}) to compute $M'X$ for any vector $X \in \mathcal A^n$ in time $O(|\mathcal B|+n)=O(n \log n)$.
This algorithm can be thought of as expressing $M'$ as $L_n S L_n$ where $S$ is a~sparse $\{-1,0,1\}$-matrix with at~most~$4|\mathcal B|$ nonzero entries and $L_n$ is the $n \times n$ full lower-triangular $\{0,1\}$-matrix.
(We are only adding and subtracting entries of~$X$, so $\mathcal A$ does not need to be a~ring.)
We can thus multiply $M'$ with any $n \times n$ matrix over~$\mathcal A$ in time $O(n^2 \log n)$, and the identity $M = PM'P^T$ implies that the same holds for~$M$. 

In contrast, \cite[Theorem 5]{twin-width5} crucially needs that the matrix $M$ itself has bounded twin-width (when seen as an ordered graph).
The algorithm indeed starts by approximating a~contraction sequence; this is known for ordered graphs \emph{but not for graphs}.
This builds upon~\cite{twin-width4} and is fairly involved and impractical.
The next step is to invoke an algorithm answering first-order queries~\cite{GajarskyPPT22}, and is even less practical.
Our new algorithm is simpler: the computation of the signed tree model and the matrix-multiplication algorithm (steps 1 and 3) are around 15--20 lines of code, and the conversion to an interval biclique partition only uses elementary computational geometry and tree operations.

It is also more general.
For instance, \cref{thm:mat-mult} allows one to compute $M_1 M_2 \ldots M_{k-1} M_k$, as $M_1(M_2(\ldots(M_{k-1} M_k)\ldots))$, where each $M_i$ is an adjacency matrix of an $n$-vertex graph, in time $O(k n^2 \log n)$.
This is not achievable by \cite[Theorem 5]{twin-width5} even if the matrices $M_i$ were of bounded twin-width, as the twin-width bound of intermediate products increases with~$k$, making the dependence in~$k$ much worse. 

\medskip

We already mentioned that graphs of bounded merge-width admit sparse signed tree models.
Those can be computed efficiently from so-called \emph{construction sequences} witnessing low merge-width; see~\cref{lem:cstr-seq-to-stm}.
Note that a~construction sequence has length~$O(n^2)$ (and sequences of length $O(n)$ and optimal width always exist).
Thus, as a~consequence of~\cref{lem:cstr-seq-to-stm,cor:apsp-gen}: 

\begin{theorem}\label{cor:merge-seq-to-apsp}
  There is an algorithm that, given a~construction sequence of an $n$-vertex~$G$ with radius-1 width~$d$, solves \textsc{All-Pairs Shortest Path} on~$G$ in $O(dn^2 \log n)$ time. 
\end{theorem}

The good dependence in~$d$ means that on graphs of radius-1 merge-width $n^{o(1)}$ (called \emph{almost bounded}) given with a~corresponding construction sequence, \textsc{All-Pairs Shortest Path} can be solved in time $n^{2+o(1)}$.
(These classes are conjectured to comprise every monadically dependent class.)

We also use \cref{thm:main} to improve the running time of the first-order model checking algorithm on graphs of bounded merge-width from cubic~\cite{merge-width} to quadratic. 

\begin{theorem}\label{thm:fomc}
  Fix $d, k \in \mathbb N$.
  There is an $r =O_1(k)$ and an algorithm that decides $G \models \varphi$ in time $O_{d,k}(n^2)$ for every $n$-vertex graph $G$ given with a~construction sequence witnessing that $G$ has radius-$r$ merge-width at~most~$d$, and every first-order sentence $\varphi$ of quantifier rank~$k$.
\end{theorem}

As we explain next, one can efficiently find almost sparse signed tree models for graphs of bounded merge-width \emph{without} witnesses.
This happens more generally to classes of (almost) linear neighborhood complexity.

\begin{figure}[!ht]
\centering
\begin{tikzpicture}[
  >=Stealth,
  node distance=7mm and 12mm,
  class/.style={
    draw, rounded corners, align=center,
    font=\footnotesize, inner sep=5pt,
    text width=32mm
  },
  note/.style={font=\scriptsize, align=center},
  incl/.style={->, line width=0.4pt},
  weak/.style={->, dashed, line width=0.4pt}
]

\node[class] (tww) {
  \textbf{bounded twin-width}\\[1mm]
  ~~~\textsc{APSP}: \hyperref[thm:apsp-tww]{\textcolor{blue}{$O(n^2\log^2 n)$}}\\
  with seq: $O(n^2)$
};

\node[class, above=of tww] (mw) {
  \textbf{bounded merge-width}\\[1mm]
  ~~~\textsc{APSP}: $O(n^2\log^3 n)$\\
  with seq: \hyperref[cor:sd-degen-to-apsp]{\textcolor{blue}{$O(n^2\log n)$}}
};

\node[class, above=of mw] (fw) {
  \textbf{bounded flip-width}\\[1mm]
  \textsc{APSP}: $O(n^2\log^3 n)$
};

\node[class, above=of fw] (symdiff) {
  \textbf{bounded symmetric difference}\\[1mm]
  \textsc{APSP}: \hyperref[thm:sym-diff]{\textcolor{blue}{$O(n^{7/3}\log^2 n)$}}\\
  even for symmetric difference $O(n^{1/3})$
};

\node[class, above=of symdiff] (sddeg) {
  \textbf{bounded sd-degeneracy}\\[1mm]
  \textsc{APSP}: $O(n^\omega\log n)$\\
  with seq: \hyperref[cor:sd-degen-to-apsp]{\textcolor{blue}{$O(n^2\log n)$}}
};

\node[class, left=of symdiff] (deg) {
  \textbf{bounded degeneracy}\\[1mm]
  \textsc{APSP}: $O(n^2)$
};

\node[class, right=of symdiff] (lnc) {
  \textbf{linear neighborhood complexity}\\[1mm]
  \textsc{APSP}: $O(n^2\log^3 n)$
};

\node[class, above=of lnc] (alnc) {
  \textbf{almost linear neighborhood complexity}\\[1mm]
  \textsc{APSP}: $n^{2+o(1)}$
};

\draw[incl] (tww) -- (mw);
\draw[incl] (mw)  -- (fw);
\draw[incl] (symdiff) -- (sddeg);
\draw[incl] (deg) -- (symdiff);
\draw[incl] (fw) -- (lnc);
\draw[incl] (lnc) -- (alnc);
\draw[incl] (lnc) -- (symdiff);

\end{tikzpicture}
\caption{Classes with sparse signed tree models, or almost signed tree models (for the two on the right column), with the best running time known for undirected unweighted \textsc{All-Pairs Shortest Path}.
  When relevant, we also indicate the running time when witnessing sequences of low width are given. 
  The results obtained in this paper are highlighted in blue.}
\label{fig:recap-apsp}
\end{figure}
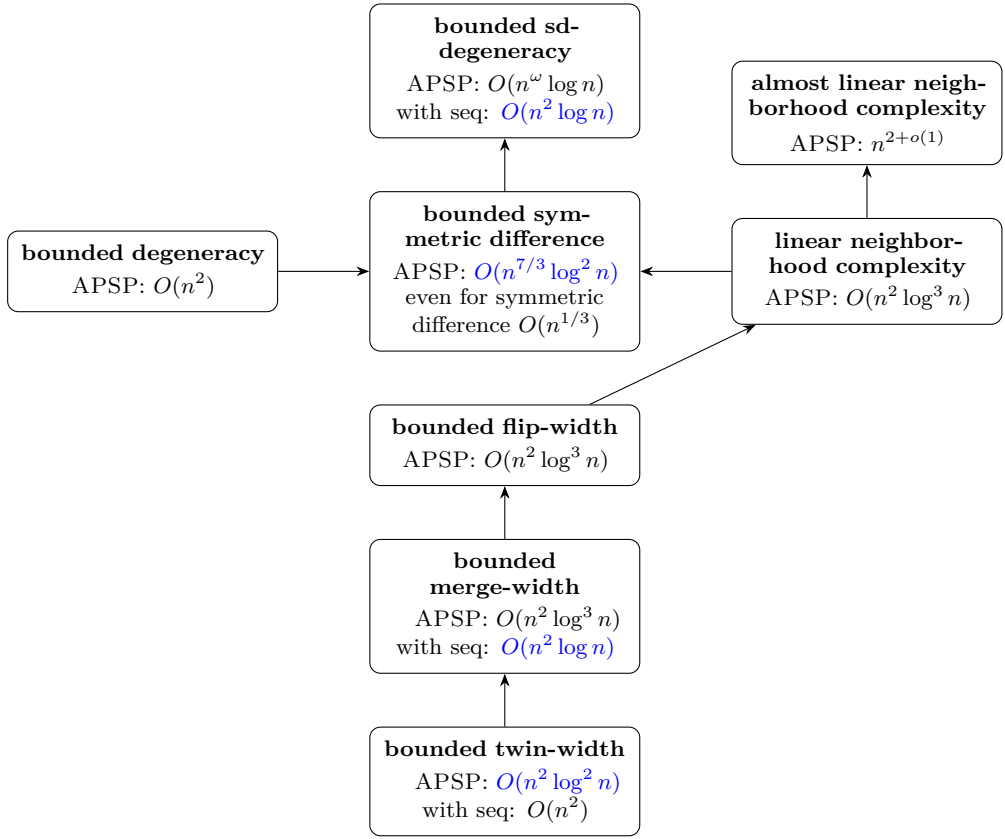

\medskip

\textbf{Related work.}
There are other sublinear encodings of graphs from structured classes.
The most relevant to the current paper works on classes of bounded VC density (also called shatter dimension).
The \emph{VC density} of~a~graph class $\mathcal C$ is the infimum of the reals $d$ such that there is a~constant $c$ satisfying that for every $G \in \mathcal C$ and for every subset $X \subseteq V(G)$, the number of distinct neighborhoods within $X$ of vertices of~$G$ is at~most~$c|X|^d$.
We further say that the class has \emph{attained VC density $d$} if this infimum is attained, that is, if the latter property holds for the infimum.
Thus classes of attained VC density~1 are precisely those of linear neighborhood complexity. 

Let us particularize a~result of Welzl on set systems whose dual has bounded VC dimension~\cite[Lemma 4.1]{Welzl88} to \emph{1-neighborhood hypergraphs}, that is, the set system $(V(G),\{N[v]:v \in V(G)\})$ for a~graph $G$.  
Observe that any such set system has as many sets as elements, and is isomorphic to its dual.
Welzl's result implies that one can order the vertex set of any $n$-vertex graph from a~class of attained VC density $d \geqslant 1$ such that if all $n$ neighborhoods are partitioned into maximal intervals along this order, then the total number of intervals is $O(n^{2-\frac{1}{d}} \log^2 n)$; see~\cite[Remark 3.12]{Anand25} for why $O(n^{2-\frac{1}{d}} \log n)$ is stated.
Several recent works have suggested efficient algorithms to compute this order (and the intervals).

\begin{lemma}[Lemma 7.4 in~\cite{Chan25} or Lemma 13 in~\cite{Duraj24}]\label{lem:int-rep}
  Let $\mathcal C$ be a~class of attained VC density~$d \geqslant 1$.
  There is an algorithm running in time $O(n^{1+\frac{1}{d}} \log n)$ that inputs an $n$-vertex graph $G \in \mathcal C$, and returns a~linear order $\prec$ on~$V(G)$ partitioning the $n$ neighborhoods into $O(n^{2-\frac{1}{d}}\log^2 n)$ maximal intervals along~it.
  In an additional $O(|E(G)|)$ time, the intervals can be computed.
\end{lemma}

We are particularly interested in the $d=1$ case, when the running time is essentially quadratic (hence listing the intervals is \emph{not} a~bottleneck) and the number of intervals is $O(n \log^2 n)$.
This can also be attained by~\cite[Lemma 3.13 and an adapted Theorem 3.6]{Anand25}.

A~few remarks are in order.
There is no mention of VC density in~\cite{Duraj24,Anand25}, solely of VC dimension.
However the property \emph{VC dimension $d$} is used only through its weaker (by the Sauer-Shelah lemma) requirement of \emph{attained VC density $d$}.
This hides a~bit the strength of the results therein as most classes of bounded VC dimension have very low attained VC density, namely 1 or 2.
In the statement of~\cite[Lemma 13]{Duraj24}, the tilde over the first big-oh is missing.
Its proof and that of \cite[Lemma 7.4]{Chan25} reveal a~$\log^2 n$ multiplicative factor (similar to \cite[Lemma 3.13]{Anand25}).
The running time in~\cite[Lemma 13]{Duraj24} is stated and repeated in the proof with some $\Tilde{O}$, suggesting at least one log factor, but it is unclear where this log factor comes from.
The bottleneck of the algorithmic consequence for \textsc{All-Pairs Shortest Path} anyway is the number of intervals.

The output of~\cref{lem:int-rep} is a~particular interval biclique representation where every biclique is a~star; one side of the biclique is a~singleton.
Thus combining this lemma with \cite[Theorem 6.2]{twin-width3}, one gets:

\begin{theorem}\label{thm:lin-neigh-comp}
  Let $\mathcal C$ be a~class of unweighted graphs of attained VC density~1.
  Then, \textsc{All-Pairs Shortest Path} can be solved in time $O(n^2 \log^3 n)$ on~$n$-vertex graphs of $\mathcal C$.
\end{theorem}

The previous theorem applies to classes of bounded merge-width and of bounded flip-width, which have linear neighborhood complexity~\cite{Bonamy25}.
Following the analysis of Welzl's arguments in the appendix of~\cite{BonnetDSZ25}, one can show \cref{lem:int-rep} for (genuine) VC density at~most~$d$.
The number of intervals is then $O(n^{2-\frac{1}{d}+o(1)})$, resulting in an $O(n^{2+o(1)})$-time algorithm for \textsc{All-Pairs Shortest Path} in classes of VC density~1.
We recall that the latter classes are believed to contain every monadically dependent class.

We finally observe that \cref{lem:int-rep} and its counterpart for (genuine) VC density at~most~$d$ also yield signed tree models (on combs) with $O(n^{2-\frac{1}{d}} \log^2 n)$ and $n^{2-\frac{1}{d}+o(1)}$ transversal pairs, respectively.

\medskip

\textbf{Open questions and further work.}
We believe that the study of signed tree models is promising, and as demonstrated in the paper, can lead to algorithmic improvements, even when those tree models are not a~priori given.
Four kinds of questions come to mind:
\begin{itemize}
  \item On which classes can we efficiently compute sparse(-ish) signed tree models? 
  \item Can we further improve \textsc{Shortest Path} algorithms based on signed tree models?
  \item Can other problems (than the distance-based ones mentioned in the abstract and those based on matrix multiplication) be solved efficiently given sparse(-ish) signed tree models? Via interval biclique partitions?
  \item How expressive are sparse signed tree models?
\end{itemize}

We exemplify with some concrete questions.
The first item and our paper suggest, by order of increasing difficulty: 

\begin{question}\label{q:comput-sstm}
  Can sparse signed tree models be computed in time $O(n^2 \log n)$ in $n$-vertex graphs of bounded twin-width? bounded merge-width? flip-width? symmetric difference? sd-degeneracy?
\end{question}

A~positive answer to~\cref{q:comput-sstm} implies, by~\cref{cor:apsp-gen}, an equally fast \textsc{All-Pairs Shortest Path} algorithm.
For twin-width, in particular, we ask if the $O(n^2)$-time algorithm with a~contraction sequence~\cite{bannach2024dag} can be attained without witness. 

\begin{question}\label{q:tww-apsp}
  Does \textsc{All-Pairs Shortest Path} admit an~$O(n^2)$-time algorithm on graphs of bounded twin-width?
\end{question}

\Cref{thm:apsp-tww} is a~$\log^2 n$ factor away, and using~\cite{Anand25} incurs a~$\log^4 n$ factor.

\begin{question}\label{q:positivizing}
  Do graphs with signed tree models with $p(n)$ transversal pairs admit tree models with $O(p(n))$ transversal edges?
  If so, can the latter be computed from the former efficiently? in time $O((n+p(n)) \log n)$? 
\end{question}

\Cref{q:positivizing} can first be asked with $O(p(n) \log n)$ transversal edges, as we currently have a~$\log^2 n$ blow-up.
We also wonder if the $\log n$ blow-up from signed tree models sparsity to the size of DAG compressions (second item of~\cref{thm:main}) can be improved.

\begin{question}\label{q:dag-compression}
  Do graphs with signed tree models with $p(n)$ transversal pairs admit DAG compressions of size $O(p(n))$?
  If so, can the latter be computed from the former efficiently? in time $O((n+p(n)) \log n)$? 
\end{question}
A~positive answer to~\cref{q:positivizing} or \cref{q:dag-compression} would shave a~$\log n$ factor in the above \textsc{All-Pairs Shortest Path} algorithms.

In light of \cref{cor:sd-degen-to-apsp}, efficiently approximating the sd-degeneracy of a~graph with a~relatively low ratio would have interesting consequences.
\begin{question}\label{q:approx-sd-degen}
  Is there a~polynomial algorithm (ideally, running in time $n^{2+o(1)}$) that given an $n$-vertex graph $G$ of sd-degeneracy at most~$d$, outputs an sd-degeneracy sequence of~$G$ of width at most $O_d(1)$? at most $O_d(1) \cdot n^{o(1)}$? at most $O_d(n^{1/3})$?   
\end{question}

By a~direct counting argument, every class admitting sparse signed tree models is factorial, i.e., its number of $n$-vertex graphs is $2^{O(n \log n)}$.
Is the converse true among hereditary classes?

\begin{question}\label{q:factorial-classes}
  Does every hereditary factorial class admit sparse signed tree models?  
\end{question}  
Permutation graphs and 2-track interval graphs are candidate classes to negatively resolve~\cref{q:factorial-classes}, but it would require new arguments to rule out the existence of a~sparse signed tree model, beyond naive counting.

We conclude with another perspective, that of dynamic algorithms.
The construction of a~signed tree model can be costly, but after it is done, \textsc{Single-Source Shortest Path} (and possibly other problems) can be solved in sublinear time (for dense graphs) in their number of edges.
We remark that updating the signed tree model when edges are inserted or deleted can simply be performed by adding transversal edges or anti-edges at the corresponding pairs of leaves.
The signed tree model can then be rebuilt (from scratch) after $p(n)$ edge edits.
For instance, using~\cref{thm:main-sssp,thm:vers-tww-to-stm} with $p(n)=O(n \log n)$, one can show:

\begin{theorem}\label{thm:tww-dynamic-sssp}
  Let $\mathcal C$ be a~class of bounded twin-width.
  There is a~fully dynamic Las Vegas algorithm for \textsc{Single-Source Shortest Paths} on graphs guaranteed to stay in $\mathcal C$, with $O(n)$ amortized update time, and $O(n \log n)$ worst-case query time.
\end{theorem}

The initialization time is $O((m+n) \log n)$ for an $n$-vertex $m$-edge graph.
Actually, \cref{thm:tww-dynamic-sssp} only needs the weaker guarantee that the graph is in $\mathcal C$ every $p(n)$ steps.
The amortized update time is in fact $O(m/n)$ where $m$ is the maximum number of edges after every $p(n)$ steps.
Also, if we know that the total number of edits is $O(p(n))=O(n \log n)$, then the update time becomes $O(1)$ worst-case.

In \emph{general} unweighted graphs, the existence of a~non-trivial (that is, with $o(n^2)$ update and query time) fully dynamic algorithm was open for a~couple of decades, until relatively recently when a~Monte Carlo (resp.~deterministic) algorithm with $O(n^{1.933})$ (resp.~$O(n^{1.969})$) worst-case update and query time was found~\cite{BrandK23}.

We note that similar schemes to \cref{thm:tww-dynamic-sssp} can be performed by maintaining an interval biclique partition.
See~\cite{Anand25} when the latter is given by~\cref{lem:int-rep}.

\subparagraph*{Organization of the paper.}
In~\cref{sec:prelim}, we give the necessary background on signed tree models, DAG compression, symmetric difference, sd-degeneracy, twin-width, and merge-width.
In~\cref{sec:sssp}, we establish the conversion algorithm of~\cref{thm:main}, and deduce \cref{thm:main-sssp}.
In~\cref{sec:tww-sd-degen}, we show \cref{thm:vers-tww-to-sd}, which by~\cref{cor:sd-degen-to-apsp}, implies~\cref{thm:apsp-tww}: a~fast \textsc{All-Pairs Shortest Path} algorithm for graphs of bounded twin-width.
With the same algorithm (and different parameters), we show \cref{thm:sym-diff-sd-degen} in~\cref{sec:sd-sd-degen}, which implies~\cref{thm:sym-diff}: an \textsc{All-Pairs Shortest Path} algorithm for graphs of moderate symmetric difference. 
In \cref{sec:mat-mult}, we present the fast matrix-multiplication algorithm when one of the two matrices is an adjacency matrix of a~graph of bounded twin-width, \cref{thm:mat-mult}.
Finally, in~\cref{sec:mw}, we show~\cref{thm:fomc}, which improves the first-order model checking algorithm on graphs of low merge-width given with a~witness.

\section{Preliminaries}\label{sec:prelim}

If $i$ and $j$ are integers, $[i,j]$ denotes the set of integers that are at least~$i$ and at~most~$j$, and $[i]$ is a~shorthand for $[1,i]$.
We will make the abuse of notation that ``$x$ is $O(x_1)=O(x_2)=\ldots=O(x_h)$'' means that $x \in O(x_1)$, $x_1 \in O(x_2)$, $\ldots$, and $x_{h-1} \in O(x_h)$.
We use $\log$ and $\ln$ for the logarithms in base 2 and $e$, respectively.

All trees considered in this paper are rooted. We use the notation $\preceq_T$ to denote the ancestor-descendant relation between nodes in a~rooted tree $T$: $u\preceq_T v$ ($u\prec_T v$) says that $u$ is an ancestor (a strict ancestor) of $v$ in $T$. The subtree of $T$ rooted at a~node $v$ is denoted as $T_v$. The set of leaves in $T$ is written as $L(T)$, which extends to $L(T_v)$ to denote the set of leaves of the subtree $T_v$.
A~rooted tree is a~\emph{binary} tree if every internal node has at most two children and it is \emph{full} if every internal node has exactly two children.

In a~universe~$U$, a~family of subsets $\Fc \subseteq 2^U$ is called \emph{laminar}
if no two sets in~$\Fc$ properly intersect, i.e., for any two $X,Y \in \Fc$, either $X \cap Y = \emptyset$, $X \subseteq Y$, or $Y \subseteq X$.
As long as~$\Fc$ does not contain the empty set, any laminar family~$\Fc$ naturally defines a~rooted forest,
where nodes are elements of~$\Fc$, and~$X$ is an ancestor of~$Y$ if and only if~$Y \subseteq X$.
We call it the \emph{inclusion forest} of~$\Fc$.

\subsection{Algorithms on rectangles}
We will use a~few well-known geometric algorithms manipulating rectangles, presented in this section.
While we state them here for rectangles in the plane,
they of course also apply to `discrete rectangles' $\{a,\dots,b\} \times \{c,\dots,d\}$ in the `discrete plane'~$[n] \times [n]$,
which is how we will use them.

Firstly, we will manipulate laminar families of rectangles.
\begin{lemma}\label{lem:rectangle-inclusion-tree}
    Given a~laminar family~$\Rc$ of~$n$ rectangles, one can compute the inclusion forest of~$\Rc$ in time~$O(n \log n)$.
\end{lemma}
The following data structure is used to prove \cref{lem:rectangle-inclusion-tree}:
\begin{theorem}[Mortensen~\cite{mortensen2006range}]\label{thm:dyn-rectangle-query}
    There is a~data structure representing a~point set~$P$ in the plane, allowing the following operations:
    \begin{enumerate}
        \item adding or deleting a~point in~$P$ in time~$O(\log |P|)$, and
        \item given a~rectangle~$R$, computing the points in $P \cap R$ in time~$O(\log |P| + \abs{P\cap R})$.
    \end{enumerate}
\end{theorem}

\begin{proof}[Proof of \cref{lem:rectangle-inclusion-tree}]
    Let $R_1 < \dots < R_n$ be the rectangles of~$\Rc$ ordered by increasing area.
    Pick an arbitrary point~$x_i$ inside~$R_i$.
    For $i < j$, the laminarity of~$\Rc$ implies that $R_i \subset R_j$ if and only if $x_i \in R_j$.
    Thus the ancestors of~$R_i$ in the inclusion tree are exactly $\{R_j: i<j \text{ and } x_i \in R_j\}$,
    and the parent of~$R_i$ is the smallest (by area) of these ancestors.

    Using this criterion, the parents can be computed in linear time by the following algorithm.
    Create a~point set~$P$ with the data structure given by~\cref{thm:dyn-rectangle-query}, initially set empty.
    We maintain the following invariant at the beginning of the $i$th round:
    \begin{equation}
        x_j \in P \quad \iff \quad \text{$j < i$ and the parent of~$R_j$ is not one of~$R_1,\dots,R_{i-1}$.}
    \end{equation}
    In the $i$th round, first compute~$R_i \cap P$.
    The invariant implies that if $x_j \in R_i \cap P$, then~$R_i$ is the parent of~$R_j$, which we record.
    Then~$P$ is updated by removing~$R_i \cap P$ and adding~$x_i$.
    After the $n$th round, all parent--child relations have been recorded.
    The points remaining in~$P$ correspond to the inclusion-wise maximal rectangles of~$\Rc$, i.e., the roots of the inclusion forest of~$\Rc$.

    Sorting the rectangles by area at the beginning of the algorithm takes time~$O(n \log n)$.
    Each point~$x_i$ is added and removed from~$P$ only once, taking~$O(n \log n)$ time in total.
    Each point~$x_j$ is returned by a~query~$P \cap R_i$ only once, since it is removed from~$P$ immediately afterwards.
    Thus the sum of query sizes $\sum_{i=1}^n |P \cap R_i|$ is at most~$n$, and the~$n$ queries also take~$O(n \log n)$ time in total.    
\end{proof}

We will also need to compute the complement of families of pairwise-disjoint rectangles.
This is proved as a~subroutine in~\cite{Rezende1989rectilinear}.
\begin{lemma}[Rezende, Lee, Wu, {\cite[Lemma~3 and~5]{Rezende1989rectilinear}}]\label{lem:rectangles-complement}
    Given a~family~$\Rc$ of~$n$ disjoint rectangles in the plane,
    one can compute in~$O(n \log n)$ a~partition of the complement of~$\Rc$ into~$O(n)$ rectangles.
\end{lemma}
Note that the disjointness assumption is crucial:
one can construct a~grid-like family of~$2n$ intersecting rectangles with roughly~$n^2$ connected components in the complement.

\subsection{Signed tree models}

\emph{Signed tree models} were introduced by~\cite{bonnet2024signedtree} as a~generalization of tree models for twin-width~\cite{BonnetNMST24,twin-width3}.
A~pair of vertices of a~rooted tree $T$ is a~\emph{transversal pair} of $T$ if neither one of the pair is an ancestor of the other.
We say that two transversal pairs $u_1v_1$ and $u_2v_2$ \emph{cross} if one of $u_1, v_1$ is a~strict ancestor of one of $u_2, v_2$, and vice versa. 
A \emph{signed tree model} is a~triple $\Tc=(T, A(T), B(T))$ consisting of a~full binary tree $T$, and two disjoint sets of transversal pairs $A(T)$ and $B(T)$ of $T$ such that $A(T)\cup B(T)$ does not contain any crossing pair.
The transversal pairs of $A(T)$ and $B(T)$ are called \emph{transversal anti-edges} and \emph{transversal edges} respectively, or sometimes \emph{negative} and \emph{positive} transversal pairs.
(The mnemonic is that $A$ stands for Anti-biclique and $B$ stands for Biclique.) 

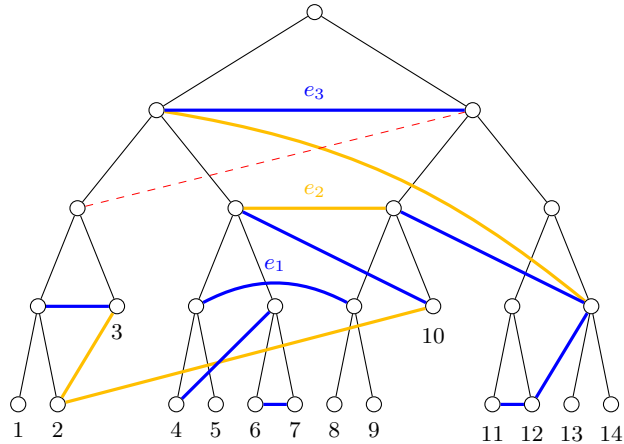
\begin{figure}[!ht]
  \centering
  \begin{tikzpicture}[%
    scale=1.3,
    level distance=10mm,
    level 1/.style={sibling distance=32mm},
    level 2/.style={sibling distance=16mm},
    level 3/.style={sibling distance=8mm},
    level 4/.style={sibling distance=4mm},
    every node/.style={draw, circle, inner sep=2pt, font=\footnotesize},
    edge from parent path={(\tikzparentnode) -- (\tikzchildnode)}
  ]
  \node (root) {}
    child {node (a) {}
      child {node (b) {}
        child {node (c) {}
          child {node (c1) [label=below:{1}] {}}
          child {node (c2) [label=below:{2}] {}}
        }
        child {node (d) [label=below:{3}] {}}
      }
      child {node (e) {}
        child {node (f) {}
          child {node (f1) [label=below:{4}] {}}
          child {node (f2) [label=below:{5}] {}}
        }
        child {node (g) {}
          child {node (g1) [label=below:{6}] {}}
          child {node (g2) [label=below:{7}] {}}
        }
      }
    }
    child {node (h) {}
      child {node (i) {}
        child {node (j) {}
          child {node (j1) [label=below:{8}] {}}
          child {node (j2) [label=below:{9}] {}}
        }
        child {node (k) [label=below:{10}] {}}
      }
      child {node (l) {}
        child {node (m) {}
          child {node (m1) [label=below:{11}] {}}
          child {node (m2) [label=below:{12}] {}}
        }
        child {node (n) {}
          child {node (n1) [label=below:{13}] {}}
          child {node (n2) [label=below:{14}] {}}
        }
      }
    };

    \foreach \i/\j/\b in {a/n/15, c2/k/0, c2/d/0}{
      \draw[very thick, amber] (\i) to [bend left=\b] (\j);
    }

    \draw[very thick, amber] (i) to node[midway, above, inner sep=1pt, draw=none] {$e_2$} (e);
    
    \foreach \i/\j/\b in {e/k/0, g1/g2/0, m1/m2/0, m2/n/0, f1/g/0, n/i/0, c/d/0}{
       \draw[very thick, blue] (\i) to [bend left=\b] (\j);
     }

\draw[very thick, blue] (f) to [bend left=28] node[midway, above, inner sep=1pt, draw=none] {$e_1$} (j);
\draw[very thick, blue] (a) to [bend left=0] node[midway, above, inner sep=1pt, draw=none] {$e_3$} (h);

    \draw[very thin, dashed, red] (b) -- (h) ; 
\end{tikzpicture}
  \caption{A~signed tree model of a~14-vertex graph, with~$A(T)$ in amber and~$B(T)$ in blue.
    The topmost amber edge and the dashed red edge cross (so the latter could not be a~transversal pair).
    Vertex 8 is adjacent to 4 because 4,8 is covered by the blue edge $e_1$, and to 2 because 2,8 is covered by the blue edge $e_3$, but 8 is not adjacent to 7 because 7,8 is covered by the amber edge $e_2$.}
\label{fig:signed-tree-model}
\end{figure}

Consider a~signed tree model $\Tc=(T,A(T),B(T))$. For a~pair $u,v$ of vertices of $T$, we say that a~transversal pair $u'v'$ \emph{covers} the pair $u,v$ if $u'\preceq_T u$ and $v'\preceq_T v$ and there is no transversal pair $u''v'' \neq u'v'$ with $u'\preceq_T u''\preceq_T u$ and $v'\preceq_T v'' \preceq_T v$.
The graph $G_{\Tc}$ of the signed tree model $\Tc$ is defined as the graph on the vertex set $L(T)$ such that $u,v\in L(T)$ are adjacent in~$G_{\Tc}$ if and only if the pair $u,v$ is covered by a~positive transversal pair in~$\Tc$.

We choose to forbid transversal pairs from being loops in signed tree models.
This is not a~significant restriction:
\begin{lemma}\label{lem:tree-model-remove-loops}
    Any signed tree model with loops can be transformed into an equivalent signed tree model without loops in linear time.
    The only new transversal pairs created in this transformation form a~matching;
    in particular, at most~$n$ transversal pairs are added.
\end{lemma}
\begin{proof}
    For each node~$t$ with children~$t_1,t_2$, if there is not already a~transversal pair between~$t_1,t_2$,
    then look for the first ancestor of~$t$ with a~loop (possibly~$t$ itself),
    and add~$t_1t_2$ as transversal pair of the same type (positive or negative) as this loop.
    This is easily implemented in linear time by a~single top-down pass.
\end{proof}

If there are only positive transversal pairs in~$\Tc$, i.e., $A(T) = \varnothing$, then~$\Tc$ is instead called a~\emph{positive tree model}.
This coincides with the notion proposed in the context of twin-width~\cite{BonnetNMST24,twin-width3}.

\subsection{Distance models and DAG compression}
Bannach, Marwitz, and Tantau~\cite{bannach2024dag} proposed \emph{DAG compression} to improve shortest-path algorithms for bounded twin-width graphs.

Coming from tree models, their idea can be explained as follows.
In signed tree models, the non-crossing condition is crucial:
without it, there may be two minimal edges covering the same pair~$u,v$, one amber and one blue,
and there is then no meaningful way to choose whether~$u,v$ should be adjacent in the corresponding graph.
In positive tree models however, this is not an issue: $u,v$ are adjacent whenever they are covered by a~transversal edge,
and this is well-defined even with crossing edges.
Bannach, Marwitz, and Tantau go one step further by replacing the tree in the tree model by a~directed acyclic graph (DAG).

Formally, a~DAG compression $(D,B)$ of a~graph~$G$ consists of a~DAG~$D$, whose sinks are exactly the vertices of~$G$, and a~set~$B$ of compressed edges between vertices of~$D$ (comparable to transversal edges of tree models), such that~$uv$ is an edge in~$G$ if and only if there is some compressed edge $xy \in B$ with directed paths $x \to u$ and $y \to v$ in~$D$.
The size of the DAG compression is the total number of vertices and edges, i.e., $|V(D)| + |E(D)| + |B|$.

Let~$G_1,G_2$ be two weighted digraphs, and $X \subseteq V(G_1) \cap V(G_2)$ a~set of shared vertices.
The graphs~$G_1,G_2$ are \emph{distance equivalent on~$X$} if for all~$x,y \in X$, $\dist_{G_1}(x,y) = \dist_{G_2}(x,y)$.
We say that~$H$ is a~\emph{distance model} for~$G$ if $V(G) \subseteq V(H)$, and~$G,H$ are distance equivalent on~$V(G)$.
The point of this definition is that~$H$ may have significantly fewer edges than~$G$, in which case  distances in~$G$ can be computed faster by computing them in~$H$.

\begin{theorem}[Bannach,Marwitz,Tantau~\cite{bannach2024dag}]\label{thm:DAG-comp}
    Let~$G$ be an unweighted (di)graph.
    Given a~DAG compression of size~$m$ for~$G$, one can in~$O(m)$ time compute a~distance model with 0--1 weights for~$G$ of size~$O(m)$.
\end{theorem}
Their construction is beautifully simple: given the DAG compression $(D,B)$, one makes two copies of~$D$ joined on~$V(G)$, one `above~$V(G)$' with edges directed downward toward~$V(G)$, and one `below~$V(G)$' with edges directed downward away from~$V(G)$.
These two copies of~$D$ have all edges with weight~0.
Then, each directed compressed edge~$xy$ becomes a~directed edge of weight~1 going from the bottom copy of~$x$ to the top copy of~$y$.

\subsection{Interval Biclique Partitions}
The third graph representation we use, also motivated by shortest-paths algorithms, is \emph{interval biclique partitions}, introduced in~\cite{twin-width3}.

An interval biclique partition (IBP)~$(\Bc,<)$ of a~graph~$G$ consists of a~linear ordering~$<$ of the vertex set~$V(G)$,
and a~partition~$\Bc$ of the edge set~$E(G)$ such that each part $B \in \Bc$ is a~biclique $(X,Y)$ for two disjoint intervals~$X,Y$ of~$(V(G),<)$.
For algorithmic purposes, we assume that~$V(G)$ is identified with~$[n]$ so that~$<$ is the natural ordering.
Then, each biclique $[a,b], [c,d]$ is represented by the tuple $(a,b,c,d)$.
The size of the interval biclique partition is thus~$O(|\Bc|)$.

In~\cite[Section~6]{twin-width3}, it is shown that interval biclique partitions of linear size
allow to solve \textsc{Single-Source Shortest Path} in $O(n \log n)$ time.
We will in fact give an alternative proof of this fact through DAG compressions in \cref{lem:IBP-to-positive-model}.

\subsection{Symmetric difference and sd-degeneracy}

We define the \emph{symmetric difference} of two vertices $u, v$ in a~graph $G$ as \[\sd_G(u,v) := |(N_G(u) \setminus \{v\}) \triangle (N_G(v) \setminus \{u\})|.\]
The \emph{symmetric difference} $\sd(G)$ of a~graph $G$ is the least nonnegative integer~$d$ such that for every induced subgraph $H$ of~$G$ with at least two vertices, there are $u \neq v \in V(H)$ such that $\sd_H(u,v) \leqslant d$.

An~\emph{sd-degeneracy sequence} of an $n$-vertex graph $G$ is a~list of pairs of vertices of $G$: $(u_1, v_1), \ldots, (u_{n-1},v_{n-1})$ such that $u_i \neq v_i$ for every $i \in [n-1]$, $V(G) = \{u_1, u_2, \ldots, u_{n-2},$ $u_{n-1}, v_{n-1}\}$, and there is no $i < j$ such that $u_i \in \{u_j, v_j\}$.
It should be thought of as a~list of pairs where the first element of the pair gets removed but the second remains, and eventually a~single vertex remains.
The \emph{width} of the sd-degeneracy sequence $(u_1, v_1), \ldots, (u_{n-1},v_{n-1})$ is defined as $\max_{i \in [n-1]} \sd_{G - \{u_1, u_2, \ldots, u_{i-1}\}}(u_i,v_i)$.
The \emph{sd-degeneracy} of~$G$ (where 'sd' stands for symmetric difference) is the least integer~$d$ such that $G$ admits an sd-degeneracy sequence of width~$d$.

Contrary to most graph parameters (and in particular, to symmetric difference), sd-degeneracy is not \emph{hereditary}: there are graphs of sd-degeneracy~1 containing induced subgraphs of arbitrarily large sd-degeneracy~\cite[Proposition 1.1]{bonnet2024signedtree}.
However, an sd-degeneracy sequence of low width is still valuable information, as it can be, for instance, turned into a~signed tree model with few transversal pairs.

\begin{lemma}[essentially Lemma 3.1 of~\cite{bonnet2024signedtree}]\label{lem:sd-degen-to-stm}
  There is an $O(dn+|E(G)|)$-time algorithm that, given an sd-degeneracy sequence $(u_1, v_1), \ldots, (u_{n-1},v_{n-1})$ of~$G$ of width~$d$, outputs a~signed tree model of~$G$ with at~most~$(d+1)(n-1)$ transversal pairs.  
\end{lemma}
\begin{proof}
  We sketch a~proof because \cite[Lemma 3.1]{bonnet2024signedtree} does not mention running times.

  We initialize the signed tree model to $(T,A(T),B(T))$ with $T$ a~forest of isolated roots, and $A(T) = B(T) = \emptyset$. 
  For $i$ from 1 to $n-1$, let $G_i := G - \{u_1, u_2, \ldots, u_{i-1}\}$.
  In time $O(\deg_{G_i}(u_i)+d)=O(\deg_G(u_i)+d)$, we compute within $N_{G_i}(u_i) \cup N_{G_i}(v_i) \setminus \{u_i,v_i\}$ the sets
  \[A_i := N_{G_i}(v_i) \setminus N_{G_i}[u_i]\text{~and~}B_i := N_{G_i}(u_i) \setminus N_{G_i}[v_i].\]
  The time bound holds since $|(N_{G_i}(u_i) \cup N_{G_i}(v_i)) \setminus \{u_i,v_i\}| \leqslant |N_{G_i}(u_i)|+d$, due to \emph{width~$d$}.
  We add $u_ia$ to $A(T)$ for every $a \in A_i$, and add $u_ib$ to $B(T)$ for every $b \in B_i$.
  We add $u_iv_i$ to $A(T)$ if $u_i$ and $v_i$ are not adjacent, and to $B(T)$ otherwise.
  This takes $O(d)$ time.
  Finally, we make a~common parent to $u_i$ and $v_i$ in~$T$.
  This parent becomes the new $v_i$.
  We remove $u_i$ from $G$ in time $O(\deg(u_i))$.

  The overall running time is \[\sum_{i \in [n-1]} O(\deg(u_i)+d) = \sum_{i \in [n-1]} O(\deg(u_i)) + \sum_{i \in [n-1]} O(d) = O(|E(G)|+dn),\] where $\sum_{i \in [n-1]} O(\deg(u_i))=O(|E(G)|)$ exploits the fact that the vertices $u_i$ are pairwise distinct.
  The claimed upper bound on the number of transversal pairs holds as every pair $(u_i,v_i)$ incurs at~most~$d+1$ transversal pairs. 
  It is not difficult to see that the resulting $(T,A(T),B(T))$ is a~signed tree model of~$G$.
  The details can be found in~\cite[Lemma 3.1]{bonnet2024signedtree}.
\end{proof}

As a~consequence of~\cref{lem:sd-degen-to-stm,cor:apsp-gen}: 

\begin{corollary}\label{cor:sd-degen-to-apsp}
  There is an algorithm that, given an sd-degeneracy sequence of~$G$ of width~$d$, solves \textsc{All-Pairs Shortest Path} on~$G$ in $O(dn^2 \log n)$ time. 
\end{corollary}

It is NP-complete to decide if a~graph has sd-degeneracy at~most~1 and co-NP-complete to determine if its symmetric difference is at~most~8 \cite{bonnet2024signedtree}, while the existence of non-trivial approximation algorithms appears to be open.
However, by a~greedy algorithm, one can find in $O(n^4)$ time an sd-degeneracy sequence of width $d$ in an $n$-vertex graph of \emph{symmetric difference} at~most~$d$.
In~\cref{thm:sym-diff-sd-degen} we provide a~faster $O(\frac{n^3}{d^2})$-time algorithm loosening the width to~$O(d \log n)$.
We instantiate it with $d=O(n^{1/3})$ as this balances its running time with that of~\cref{cor:sd-degen-to-apsp}, and establishes~\cref{thm:sym-diff}.

\subsection{Twin-width}

We will not explicitly need the definitions of twin-width and versatile twin-width, only a~consequence of bounded versatile twin-width.
We thus refer the interested reader to~\cite{twin-width1,twin-width2}.
In this paper only three facts on twin-width will matter.

Twin-width is a~hereditary parameter.
\begin{observation}[\cite{twin-width1}]\label{obs:tww-hered}
  For every graph $G$, every induced subgraph of $G$ has twin-width at~most~the twin-width of~$G$. 
\end{observation}

Twin-width and versatile twin-width are functionally equivalent.
\begin{lemma}[Lemma 3.8 in \cite{twin-width2}]\label{thm:tww-vers-tww}
  There is a~function $f_{\ref{thm:tww-vers-tww}}(d) = 2^{2^{O(d)}}$ such that every graph of twin-width at~most~$d$ has versatile twin-width at~most~$f_{\ref{thm:tww-vers-tww}}(d)$.
\end{lemma}

And finally:
\begin{observation}[part of the definition, see~\cite{twin-width2}]
  If an $n$-vertex graph $G$ has versatile twin-width at~most~$d$, then there are at least $h := \lfloor n/d \rfloor$ disjoint pairs $(u_1,v_1), \ldots, (u_h,v_h)$ of vertices of $G$ such that for every $i \in [h]$, $\sd_G(u_i,v_i) \leqslant d$.  
\end{observation}

\subsection{Merge-width}
Dreier and Toruńczyk~\cite{merge-width} give two equivalent definitions of merge-width, through \emph{construction sequences} and \emph{merge sequences} respectively.
The former is more appropriate in our setting.

A construction sequence is a~sequence of steps $(\Pc_1,E_1,N_1),\dots,(\Pc_m,E_m,N_m)$ constructing a~graph~$G$.
The vertex set $V = V(G)$ is fixed throughout the process.
At each step, $\Pc_i$ is a~partition of~$V$, and $E_i,N_i \subseteq \binom{V}{2}$ are of \emph{resolved edges} and \emph{non-edges} respectively
(intuitively, pairs that have already been determined to belong, respectively not belong to~$E(G)$).
Initially, $\Pc_1$ is the partition into singletons, and $E_1 = N_1 = \emptyset$.
At each step, $(\Pc_{i+1},E_{i+1},N_{i+1})$ are obtained from $(\Pc_i,E_i,N_i)$ by one of the following three operations:
\begin{description}
    \item[merge] Choose two parts $A,B \in \Pc_i$ and replace them by $A \cup B$ in~$\Pc_{i+1}$.
    \item[resolve positively] Choose two parts $A,B \in \Pc_i$ (possibly $A=B$), and set $E_{i+1} = E_i \cup (AB \setminus N_i)$,
    i.e., each pair between~$A$ and~$B$ becomes an edge, except for those that were previously determined to be non-edges.
    \item[resolve negatively] Symmetrically, choose $A,B \in \Pc_i$ and set $N_{i+1} = N_i \cup (AB \setminus E_i)$.
\end{description}
The graph constructed by this sequence is $G = (V,E_m)$.

Optionally, one may require that at the end of the process,
$E_m,N_m$ partition~$\binom{V}{2}$ (note that they are always disjoint),
and that~$\Pc_m = \{V\}$ is the trivial partition.
These two conditions can always be ensured by merging all remaining parts together, and resolving the last part~$V(G)$ negatively with itself.

Let $R_i = E_i \cup N_i$ be the set of all resolved pairs at step~$i$.
For $r \in \Nn$, the \emph{radius-$r$ width} of this construction sequence is
\[ \max_{i \in [m]} \max_{v \in V} |\Ball^r_{R_i}(v) / \Pc_i|, \]
i.e., the maximum number of parts $A \in \Pc_i$ that are at distance at most~$r$ from some vertex~$v$ in the graph of resolved pairs $(V,R_i)$.
The \emph{radius-$r$ merge-width} of~$G$, denoted~$\mw_r(G)$, is the minimum radius-$r$ width of a~construction sequence that constructs~$G$.

The trivial bound on the length of a~construction sequence is~$O(n^2)$
(beyond these, there must be several resolve operations between the exact same pair of parts, which is obviously useless).
If~$d$ is the radius-1 width of the sequence, \cite{merge-width} shows that the bound is in fact $(2d+1)n$.
We restate it algorithmically:
\begin{lemma}[{\cite[Lemma~3.4]{merge-width}}]\label{lem:short-cstr-seq}
  There is an $O(m)$-time algorithm that, given a~construction sequence $(\Pc_1,E_1,N_1), \ldots, (\Pc_m,E_m,N_m)$ of an $n$-vertex~$G$ with radius-$r$ width~$d$, computes a~construction sequence for~$G$ of length at most~$(2d+1)n$, with the same radius-$r$ width.
\end{lemma}
\begin{proof}
    The argument of~\cite[Lemma~3.4]{merge-width} is as follows:
    consider a~subsequence of a~construction sequence, consisting only of resolve operations, followed by a~merge of say~$P,Q$.
    Then any resolve of this subsequence that does not involve~$P$ or~$Q$ can be postponed after this merge, without affecting the validity or increasing the width of the construction sequence.
    Once all possible merges have been postponed in this way, there can only be~$2d$ resolve operations between two consecutive merges, unless there are some trivially useless resolve operations, i.e., repeated between the same two parts.
    Since there are~$n-1$ merges, this proves the bound.

    Algorithmically, the issue is thus to implement this postponing efficiently.
    Enumerate the merge operations as~$M_1, \ldots, M_{n-1}$, and all the parts that appear throughout the merge sequence as $P_1, \ldots, P_{2n-1}$
    (there are exactly~$2n-1$ of them:~$n$ singletons and~$n-1$ created by merges).
    For each part~$P_i$, determine the index~$d(i)$ of the merge operation~$M_{d(i)}$ that destroyed~$P_i$ by merging it with another part.
    Then a~resolve operation between~$P_i$ and~$P_j$ needs to be postponed until just before $M_{\min(d(i),d(j))}$.
    With this description, it is easy to reconstruct the modified construction sequence, and all the above can be implemented in~$O(m)$.

    At this point, if $m > (2d+1)n$, there must be some trivially useless resolve operations, which are easily detected and removed in linear time.
\end{proof}

Recall that (positive) tree models were introduced to represent bounded twin-width graphs.
Signed tree models can be used for bounded merge-width in a~similar way.
\begin{lemma}\label{lem:cstr-seq-to-stm}
    From a~construction sequence for $n$-vertex graph~$G$ with~$p$ resolve operations,
    one can construct in~$O(p+n)$ time a~signed tree model with at most~$n+p$ transversal pairs.
\end{lemma}
\begin{proof}
    The nodes of the tree model are exactly the parts appearing throughout the merge sequence:
    the singletons are the leaves (in bijection with~$V(G)$), and whenever parts~$P,Q$ are merged,
    $P \cup Q$ is added as parent of~$P$ and~$Q$.
    Positive or negative resolves between parts~$P,Q$ become positive or negative transversal pairs between the corresponding nodes.
    There cannot be crossing transversal pairs, as they would contradict the flow of time in the construction sequence.
    It is then easy to check that this is a~signed tree model for~$G$, and to construct it in~$O(p+n)$ time.

    There is only one issue: construction sequences allow resolving a~part with itself,
    hence this tree model may have loops. They can be removed with \cref{lem:tree-model-remove-loops}.
    The final tree model has~$p$ transversal pairs from resolve operations, and~$n$ from the use of \cref{lem:tree-model-remove-loops}.
\end{proof}

\section{Single-Source Shortest Path for signed tree models}\label{sec:sssp}
In this section, we prove the following.
\mainthm*

Consider a~signed tree model $\Tc = (T,A,B)$ for~$G$.
For each node~$t$ of~$T$, make an arbitrary choice of left-to-right ordering on the children of~$t$,
and call~$<_T$ the resulting left-to-right ordering on the leaves of~$T$, which are exactly the vertices of~$G$.
For simplicity, we will henceforth identify~$V(G)$ and the leaves of~$T$ with~$[n]$, so that~$<_T$ is the natural ordering on~$[n]$.

For a~node~$t$ of~$T$, let $L(t) = \{v \in V(G) : v \qgt_T t\}$ denote the set of leaves in the subtree of~$t$.
This is an interval of~$<_T$.
Now each transversal pair $uv \in A~\cup B$ can be interpreted as a~rectangle $L(u) \times L(v)$ in the adjacency matrix of~$G$, where the columns and rows are (symmetrically) ordered by $<_T$. Note that $L(u)$ and $L(v)$ are disjoint, as neither of the pair $u,v$ is a~strict ancestor of the other and we further assumed $u\neq v$ in our signed tree model.
There will also be the symmetric rectangle $L(v) \times L(u)$ for the other orientation~$vu$. We break the symmetry by choosing the rectangle $L(u)\times L(v)$ with $L(u) <_T L(v)$.
We denote by~$\Rc_A$ and~$\Rc_B$ the families of rectangles defined by the transversal pairs of~$A$ and~$B$ respectively,
and $\Rc_\Tc = \Rc_A \cup \Rc_B$.

\begin{lemma}
    The family $\Rc_\Tc$ of rectangles of a~signed tree model $\Tc$ forms a~laminar family.
\end{lemma}
\begin{proof}
    Consider $R_1,R_2 \in \Rc_\Tc$, corresponding to transversal pairs~$u_1v_1$ and~$u_2v_2$ in~$\Tc$.
    Note that the family of intervals $\{L(v):v \in V(T)\}$ defined by~$\Tc$ is itself a~laminar family, and that $L(u) \subseteq L(v)$ if and only if~$u$ is a~descendant of~$v$ in~$\Tc$.
    Thus~$L(u_1)$ and~$L(u_2)$ are either disjoint or contained in one another.
    If they are disjoint, then so are~$R_1,R_2$ and we are done.
    We thus assume that~$L(u_1), L(u_2)$ are contained in one another, and similarly for~$L(v_1), L(v_2)$.
    Up to symmetry, this leaves only two cases:
    \begin{itemize}
        \item If $L(u_1) \subseteq L(u_2)$ and $L(v_1) \subseteq L(v_2)$, then $R_1 \subseteq R_2$ as desired, and similarly when all inclusions are reversed.
        \item Otherwise, $u_1$ is a~strict descendant of~$v_1$ and~$v_2$ a~strict descendant of~$u_2$ (or vice versa).
          This contradicts the non-crossing condition of the definition of signed tree model. \qedhere
    \end{itemize}
\end{proof}

Thus we can consider the inclusion forest of~$\Rc_\Tc$, which can be computed in time $O(m \log m) = O(m \log n)$ when $|\Rc_\Tc| = m$ by \cref{lem:rectangle-inclusion-tree}.
In this inclusion forest, it is in general possible to have a~rectangle~$R$ and its parent~$R'$ of the same sign (positive or negative),
meaning that $R,R' \in \Rc_A$ or $R,R' \in \Rc_B$ is allowed.
This, however, is pointless: if this case happens, then~$R$ and the corresponding transversal pair can simply be deleted from the family of rectangles and the signed tree model $\Tc$ without changing the graph represented by~$\Tc$.
Thus we obtain the following.
\begin{observation}\label{obs:alternate}
    For any graph $G$ and a~signed tree model $\Tc$ defining $G$, there exists a~signed tree model $\Tc'$ defining $G$ in which no transversal pair is covered by a~transversal pair of the same sign. Moreover, such a~model $\Tc'$ can be constructed in time $O(m\log m)$, where $m=\abs{A(T)\cup B(T)}$.
\end{observation}

We are now ready to begin the construction of an interval biclique partition.
\begin{lemma}\label{lem:tree-model-to-IBP}
Let $\Tc=(T,A(T),B(T))$ be a~signed tree model of $G$ on the vertex set $[n]$. Then $G$ admits an interval biclique partition with $O(m)$ bicliques, which can be constructed in time $O(m \log m)$, where $m=\abs{A(T)}+\abs{B(T)}$.  
\end{lemma}
\begin{proof}
    Consider the family of rectangles $\Rc_T=\Rc_A \cup \Rc_B$ of $\Tc$, where $\Rc_A$ and $\Rc_B$ are the families of rectangles for the negative and positive transversal pairs, respectively. 
    One can compute the inclusion forest of $\Rc_T$ in time $O(m \log m)$ by~\cref{lem:rectangle-inclusion-tree}. 
    Moreover, we may assume that the sign (positive or negative) of the nodes in the inclusion forest alternate by~\cref{obs:alternate}. 
    
    Consider a~positive node and its (negative) children, each corresponding to the rectangle $R\in \Rc_B$ and $R_1, \ldots, R_\ell\in \Rc_A$. 
    Notice that $R_1, \ldots, R_\ell$ are pairwise disjoint due to the laminarity of $\Rc$. By~\cref{lem:rectangles-complement}, in time 
    $O(\ell  \log \ell)$, one can compute a~new family of rectangles $\Pc(R)$ consisting of~$O(\ell)$ rectangles which partitions $R \setminus \bigcup_{i=1}^\ell R_i$. 
    
    Let $\Pc$ be the union of $\Pc(R)$ over all $R\in \Rc_B$ and $\ell_R$ be the number of (negative) children for each positive $R\in \Rc_B$ in the inclusion forest. 
    By iterating the above procedure for all positive nodes, one can compute $\Pc$ in time $O(\sum_{R\in \Rc_B}\ell_R\log \ell_R)=O(\abs{\Rc_A}\cdot \log \abs{\Rc_A})=O(m \log m)$. We also note that the number of parts in $\Pc$ is at most $O(m)$. 

    Now, we argue that $\Pc$ is an interval biclique partition of $G$.
    Note that all rectangles in $\Pc$ are pairwise disjoint.
    Indeed, any two rectangles of $\Pc$ computed for two disjoint positive rectangles in $\Rc_B$ or the same rectangle in $\Rc_B$ are clearly disjoint. 
    Consider two positive rectangles $R,R'\in \Rc_B$ with $R\subsetneq R'$. 
    Then by the assumption that the sign of the nodes alternates in the inclusion forest, there is a~negative rectangle $R''\in \Rc_A$ such that $R\subsetneq R''\subsetneq R'$. 
    Furthermore, we can pick $R''$ to be a~child of $R'$ in the inclusion forest of~$\Rc_T$.
    If two new rectangles $P$ and $P'$ in $\Pc$ are obtained when handling $R$ and $R'$ respectively, then $P$ is fully contained in $R$ and thus contained in $R''$, whereas $P'$ is disjoint from $R''$ by construction of $\Pc(R')$.
    
    Finally, observe an ordered pair $(i,j)\in [n]^2$ with $i<_T j$ has an edge between them in $G$ if and only if $(i,j)$ is covered by some rectangle in $\Pc$. 
    Therefore, each rectangle $P=I\times J\subseteq [n]\times [n]$ in $\Pc$ seen as a~set of unordered pairs in $[n]\times [n]$ is a~set of edges in the form $(i,j)$, where $i$ and $j$ are respectively taken from the intervals $I$ and $J$. 
    That is, each rectangle is an interval biclique. By the pairwise disjointness of rectangles in $\Pc$, it follows that the family $\Pc$ is an interval biclique partition of $G$.
\end{proof}

Finally, from this interval biclique partition, we can reconstruct the other two items of \cref{thm:main},
namely a~DAG and a~positive tree model.
\begin{lemma}\label{lem:IBP-to-positive-model}
    For a~graph~$G$ with~$n$ vertices, given an interval biclique partition of~$G$ into~$p$ bicliques,
    one can compute:
    \begin{itemize}
        \item a~DAG compression of size $O(p \log n)$ for~$G$ in time~$O(p \log n)$, and
        \item a~positive tree model with $O(p \log^2 n)$ transversal edges for~$G$ in time~$O(p \log^2 n)$.
    \end{itemize}
\end{lemma}
\begin{proof}
    Let~$<$ be the ordering of~$V(G)$ used by the interval biclique partition, and denote by $B_1,\dots,B_p$ the bicliques.
    Each $B_i$ is given as $B_i = I_i \times J_i$ for some intervals~$I_i,J_i$ of~$<$.

    As basis for both the DAG compression and the positive tree model, we use the balanced binary tree~$T$ of height $h = \ceil{\log n}$ whose leaves are~$V(G)$ ordered left-to-right by~$<$,
    constructed in linear time.
    For a~node $t \in T$, recall that $L(t) \subseteq V(G)$ is the subset of leaves that are descendants of~$t$.

    Any interval $I = [a,b]$ of~$<$ can be partitioned as $I = \biguplus_{t \in S_I} L(t)$ for some subset of nodes of~$T$ with size $|S_I| \le 2h$.
    Given~$a,b$ the set~$S_I$ is computed in time~$O(h)$, by following the paths from~$a,b$ to the root,
    and adding to~$S_I$ all nodes that branch out from these paths between~$a$ and~$b$.

    Now for each biclique $B = I \times J$, we compute the sets~$S_I$ and~$S_J$, and proceed as follows depending on the desired representation:
    \begin{itemize}
        \item To obtain a~positive tree model, we simply add each edge~$st$ with $s \in S_I$ and $t \in S_J$ as a~transversal edge.
        This constitutes at~most $4 \ceil{\log n}^2$ transversal edges for each biclique.
        Note that (1) the bicliques $B_1,\dots,B_p$ are assumed to be edge-disjoint, and (2) the interval~$I$ is partitioned into $\{L(t)\}_{t \in S_I}$ which are again all disjoint,
        which implies that the transversal edges created in this way again correspond to pairwise disjoint bicliques in the graph, and in particular cannot cross.
        Therefore, this is a~positive tree model of size $O(p \log^2 n)$.
        \item To obtain a~DAG compression, we instead create new nodes~$v_I$ and~$v_J$ in the DAG
        with children~$S_I$ and~$S_J$ respectively, and add the transversal edge~$v_Iv_J$.
        This only adds $2\ceil{\log n}+1$ edges to the DAG compression for each biclique, hence $O(p \log n)$ in total. \qedhere
    \end{itemize}
\end{proof}

\Cref{thm:main} follows directly from \cref{lem:tree-model-to-IBP,lem:IBP-to-positive-model}.

\section{Sd-degeneracy sequences in graphs of bounded twin-width}\label{sec:tww-sd-degen}

The next theorem and \cref{cor:sd-degen-to-apsp} imply~\cref{thm:apsp-tww}.

\begin{theorem}\label{thm:vers-tww-to-sd}
  Let $f(d) := f_{\ref{thm:tww-vers-tww}}(d) \geqslant 1$ and fix any constant $c > 0$.
  There is an algorithm that inputs an $n$-vertex $m$-edge graph $G$ of twin-width at~most~$d$, and outputs an~sd-degeneracy sequence of width $O(f(d) \log n)$, in time $O(f(d) (m+n) \log n)$ with probability at~least~$1-n^{-c}$.
\end{theorem}

\begin{proof}
  In order to factorize this proof and that of~\cref{thm:sym-diff-sd-degen}, we define $g := f(d)$ (the \emph{good} threshold), $\gamma := 2f(d)(c+3)\ln n$ (the \emph{good-enough} threshold), and $p := \frac{1}{2f(d)}$ (the sampling probability).
  For every induced subgraph $H$ of~$G$, a~pair $u \neq v \in V(H)$ is \emph{good} (in $H$) if $\sd_H(u,v) \leqslant g$, and it is \emph{good-enough} if $\sd_H(u,v) \leqslant \gamma$.
  To ease some computations, we assume that $n$ is large enough compared to~$c$ and $d$.
  
  \subparagraph*{Algorithm.}
  We initialize the sd-degeneracy sequence with the empty list $\mathcal L$.

  While $|V(G)| > 2g$, we proceed as follows.
  We first build a~set~$X \subseteq V(G)$ by adding to it each $v \in V(G)$ independently with probability $p$.
  We then construct the partition $Q$ of $V(G)$ into the neighborhood classes toward~$X$, i.e., two vertices $u$ and $v$ are in the same part of~$Q$ if and only if $N_G(u) \cap X = N_G(v) \cap X$.

  For every non-singleton part $P$ of $Q$, we arbitrarily form $\lfloor \frac{|P|}{2} \rfloor$ disjoint pairs of vertices of~$P$.
  For every formed pair $(u,v)$, we check if $\sd_G(u,v) \leqslant \gamma$ holds and, if so, append $(u,v)$ to $\mathcal L$.
  Let $(u_1,v_1), \ldots, (u_s,v_s)$ be the appended pairs at this iteration of the while loop.
  We set $G := G-\{u_1, \ldots, u_s\}$; this finishes the body of the while loop.

  When we exit the while loop, the remaining graph has at~most~$2g < \gamma$ vertices, and we finish the sd-degeneracy sequence arbitrarily. 
  
\begin{algorithm}[t]
\caption{Las Vegas algorithm in time $O(f(d)(m+n)\log n)$ with probability $\geqslant 1-n^{-c}$.}
\label{alg:vers-tww-to-sd}
\DontPrintSemicolon
\KwIn{An $n$-vertex $m$-edge graph $G$ of twin-width at most $d$, a~constant $c>0$.}
\KwOut{An sd-degeneracy sequence of width $O(f(d)\log n)$, with $f(d) := f_{\ref{thm:tww-vers-tww}}(d)$.}

\medskip

$g \gets f(d)$                     \tcp*[r]{sd threshold of good pairs}
$\gamma \gets 2f(d)(c+3)\ln n$     \tcp*[r]{sd threshold of good-enough pairs}
$p \gets \frac{1}{2f(d)}$ \tcp*[r]{sampling probability}
$\mathcal L \gets [\,]$ \tcp*[r]{the sd-degeneracy sequence (list of pairs)}

\While{$|V(G)| > 2g$}{
  $X \gets \emptyset$ \tcp*[r]{initialize sampled set $X$}
  \ForEach{$v \in V(G)$}{
    add $v$ to $X$ independently with probability $p$\;
  }

  $Q \gets \{\}$ \tcp*[r]{partition by neighborhoods toward $X$ (dictionary)}
  \ForEach{$v \in V(G)$}{
    compute $\vec v[X] \in \{0,1\}^{X}$ where $\vec v[X](x)=\mathbf{1}[vx\in E(G)]$\;
    $Q[\vec v[X]] \gets Q[\vec v[X]] \cup \{v\}$\;
  }

  $U \gets \emptyset$ \tcp*[r]{vertices removed at the end of iteration}
  \ForEach{part $P$ of~$Q$}{
    \If{$|P|\ge 2$}{
      form arbitrarily $\left\lfloor |P|/2\right\rfloor$ disjoint pairs $(u,v)$ in $P$\;
      \ForEach{formed pair $(u,v)$}{
        \If{$\sd_G(u,v) \le \gamma$}{
          append $(u,v)$ to $\mathcal L$\;
          $U \gets U \cup \{u\}$\;
        }
      }
    }
  }
  $G \gets G - U$\;
}

\While{$|V(G)| \ge 2$}{
  pick any distinct $u,v\in V(G)$\;
  append $(u,v)$ to $\mathcal L$\;
  $G \gets G-\{u\}$\;
}

\Return{$\mathcal L$}\;
\end{algorithm}

This finishes the description of the algorithm; see~\cref{alg:vers-tww-to-sd}.
Note that while $G$ shrinks, $n$ still denotes the initial number of vertices.
In particular, $\gamma$ does not change.
In contrast, $N_G$ and $\sd_G$ are meant in the current (induced sub)graph $G$.

  \subparagraph*{Correctness and running time.}
  Every pair added to the sd-degeneracy sequence has symmetric difference at~most~$\gamma$ in the current induced subgraph.
  Thus if the algorithm terminates, it does indeed output an sd-degeneracy sequence of the claimed width.
  We also observe that each iteration is executed on an induced subgraph of the initial graph $G$, thus a~graph of twin-width at~most~$d$, and versatile twin-width at~most~$f(d)$.

  We encode $G$ as an \emph{adjacency map}, which supports edge queries in expected constant time, neighborhood traversal of~$v$ in $O(\deg(v))$ time, and deletion of~$v$ in expected $O(\deg(v))$ time.

  It takes $O(n)$ time to build~$X$.
  It takes $O(\sum_{v \in V(G)} \deg(v))=O(m)$ time to compute the partition~$Q$:
  For each $v \in V(G)$, list the neighbors of $v$ to build the vector $\vec{v}[X] \in \{0,1\}^X$ of adjacencies w.r.t.~$X$, in $O(\deg(v))$ time.
  If this vector is measured for the first time (in the current iteration of the while loop), initialize a~new part of~$Q$ as $\{v\}$ at the address $\vec{v}[X]$, and add $\vec{v}[X]$ to the dictionary of existing vectors.
  Otherwise add $v$ to the preexisting part at address $\vec{v}[X]$.
  (Note that we do not range over $X$ to implement line 11 of~\cref{alg:vers-tww-to-sd}.)

  From the partition representation, forming the pairs can be done in $O(n)$ time.
  For every formed pair, checking whether it is good-enough can be done in $O(m+n)$ time.
  Indeed, as the pairs are disjoint, we range over at~most~$2m$ edges (every edge is scanned at most twice).
  The vertex deletion $G-U$ at the end of the while loop can be done in $O(\sum_{u \in U} \deg(u))$.
  Overall, the deletions take $O(m)$ time.
  
  We are left with proving that the number of iterations of the while loop is $O(f(d) \log n)$ with probability at~least~$1-n^{-c}$.   
  We start with two important observations, focusing on a~single iteration.
  \begin{enumerate}
  \item\label{it:tww-1} For every good pair $(u, v)$, we have \[\mathbb{P}_X[u\text{~and~}v\text{~are in the same part of~}Q] \geqslant (1-p)^g = \left(1-\frac{1}{2f(d)}\right)^{f(d)} \geqslant \frac{1}{2}.\]
  \item\label{it:tww-2} For every $u \neq v$ that does \emph{not} form a~good-enough pair, it holds that \[\mathbb{P}_X[u\text{~and~}v\text{~are in the same part of $Q$}] \leqslant (1-p)^\gamma = \left(1-\frac{1}{2 f(d)}\right)^{2f(d) (c+3) \ln n}\]
    \[\leqslant \exp\left(-\frac{2f(d)(c+3)\ln n}{2f(d)}\right) = n^{-c-3}.\]
  \end{enumerate}

  Let $n'$ denote the number of vertices of the current induced subgraph handled at this iteration; so $n'>2f(d)$.
  By~\cref{thm:tww-vers-tww}, $G$ admits at least $\lfloor n'/f(d) \rfloor$ pairwise disjoint good pairs.
  Fix a~set $\mathcal P$ of $q \geqslant \lfloor n'/f(d) \rfloor$ pairwise disjoint good pairs.
  For every pair $(u,v) \in \mathcal P$, let $Z_{u,v}$ be the indicator: $Z_{u,v}=1$ if $u$ and $v$ are in the same part of~$Q$, and 0 otherwise.
  Let $Z := \sum\limits_{(u,v) \in \mathcal P} Z_{u,v}$.
  We have
  \[\mathbb E[Z] \;=\; \sum_{(u,v)\in \mathcal P}\mathbb E[Z_{u,v}] \;\geqslant\; \sum_{(u,v)\in \mathcal P} \frac{1}{2} \;=\; \frac{q}{2},\]
  where the inequality comes from~\cref{it:tww-1}. 

  \begin{claim}\label{clm:Z-tww}
    $\mathbb P[Z \geqslant \frac{q}{4}] \geqslant \frac{1}{3}$.
  \end{claim}
  \begin{proof}
  Indeed, by setting $p' := \mathbb P[Z \geqslant \frac{q}{4}]$, since $0 \leqslant Z \leqslant q$, we have $\frac{q}{2} \leqslant \mathbb E[Z] \leqslant (1-p')\frac{q}{4} + p' q$.
  Therefore $2 \leqslant (1-p') + 4p'$, hence $p' \geqslant \frac{1}{3}$.
  \end{proof}

  We say that an iteration of the while loop is a~\emph{success} if $Z \geqslant \frac{q}{4}$ holds and all the formed pairs (during this iteration) are good-enough.
  In case of a~success, at~least~$\frac{q}{4}$ good pairs of~$\mathcal P$ have both vertices in the same part of~$Q$.
  In particular, $\sum_{P \in Q} \lfloor |P|/2 \rfloor \geqslant \frac{q}{4}$, hence during such an iteration we form at~least~$\frac{q}{4}$ disjoint pairs.
  Therefore, the algorithm appends at~least~$\frac{q}{4}$ pairs to the sd-degeneracy sequence and deletes at~least~$\frac{q}{4}$ vertices at the end of the iteration.
  As $n'>2f(d)$, we have
  \[q \geqslant \left\lfloor \frac{n'}{f(d)} \right\rfloor \geqslant \frac{n'}{2f(d)}.\]
  Thus every success deletes at least $\frac{q}{4} \geqslant \frac{n'}{8f(d)}$ vertices, hence multiplies the current number of vertices by at most $1-1/(8f(d))$.

  \begin{claim}\label{clm:nb-successes-tww}
    The total number of successes is at most $8 f(d)\ln n$.
  \end{claim}
  \begin{proof}
    Let $n_s$ be the number of vertices right before the $(s+1)$st success.
    By the previous paragraph, for every $s\geqslant 0$ we have
    \[
      n_s \;\leqslant\; n_{s-1} \left(1-\frac{1}{8f(d)}\right)
      \;\leqslant\; n\left(1-\frac{1}{8f(d)}\right)^s
      \;\leqslant\; n\exp\!\left(-\frac{s}{8f(d)}\right).
    \]
    For $s \geqslant 8 f(d)\ln n$, the right-hand side is at most $n\exp(-\ln n)=1$, hence below $2f(d)$.
    Therefore, after $8 f(d)\ln n$ successes, the while loop condition $|V(G)|>2f(d)$ cannot hold any more.
    Thus there cannot be more successes.
  \end{proof}

  Let $t$ be the number of iterations of the while loop.

  \begin{claim}\label{clm:nb-iters-tww}
    $\mathbb P[t \geqslant 32(c+3)f(d)\ln n] \leqslant n^{-c-3}$.
  \end{claim}
  \begin{proof}
    Let $I_1,\ldots,I_t$ be the indicators of success: for every $i\in[t]$, $I_i=1$ if the $i$th iteration is a~success, and $I_i=0$ otherwise.
    By~\cref{clm:nb-successes-tww}, it holds that $\sum_{i\in[t]} I_i \leqslant 8 f(d)\ln n$.

    During one iteration, by~\cref{it:tww-2} and union bound over every unordered pair that is not good-enough, the probability that at~least one such pair lands in the same part of the partition $Q$ (for this iteration) is at~most
    \[
      \binom{n}{2}\cdot n^{-c-3} \;\leqslant\; \frac{n^2}{2}\cdot n^{-c-3} \;=\; \frac{1}{2}\,n^{-c-1}.
    \]
    Thus by~\cref{clm:Z-tww}, for every $i\in[t]$ and every history description $\mathcal H$ of the previous iterations, we have
    \[
      \mathbb P[I_i \mid \mathcal H]
      \;\geqslant\; \frac{1}{3} - \frac{1}{2}\,n^{-c-1}
      \;\geqslant\; \frac{1}{4},
    \]
    for large enough $n$.

    Let $t' := 32(c+3)f(d)\ln n$ and $\tau := \lfloor 8 f(d)\ln n\rfloor$.
    As $\sum_{i \in [t]} I_i \leqslant \tau$, the probability that $t \geqslant t'$ is at most the probability that a~binomial random variable $Y \sim B(t',\frac{1}{4})$ takes value at most~$\tau$.
    Since \[\mathbb E[Y] = \frac{t'}{4} = 8(c+3)f(d)\ln n \geqslant (c+3)\tau,\] we have
    $\tau \leqslant \frac{\mathbb E[Y]}{c+3} \leqslant \frac{\mathbb E[Y]}{2}$.
    By Chernoff's bound,
    \[
      \mathbb P\!\left[Y \leqslant \frac{\mathbb E[Y]}{2}\right]
      \;\leqslant\; \exp\!\left(-\frac{\mathbb E[Y]}{8}\right)
      \;\leqslant\; \exp\!\left(-(c+3)f(d)\ln n\right)
      \;=\; n^{-(c+3)f(d)}
      \;\leqslant\; n^{-c-3},
    \]
    using $f(d)\geqslant 1$.
  \end{proof}

  By~\cref{clm:nb-iters-tww}, the while loop has fewer than $32(c+3)f(d)\ln n$ iterations with probability at~least $1-n^{-c-3}$.
  Conditional on that, by the union bound and the estimate $\frac{1}{2}n^{-c-1}$ above, the probability that during at~least one iteration there exists a~pair that is not good-enough but lands in the same part of~$Q$ is at~most
  \[
    32(c+3)f(d)\ln n \cdot \frac{1}{2}\,n^{-c-1}
    \;\leqslant\; \frac{1}{2}\,n^{-c}
  \]
  for large enough $n$.
  Therefore, with probability at~least $1-n^{-c}$, the number of iterations is $O(f(d)\log n)$ (in particular, the algorithm terminates) and all formed pairs are good-enough.

  Finally, since each iteration runs in time $O(m+n)$ on an induced subgraph of the input and there are $O(f(d)\log n)$ iterations with probability at least $1-n^{-c}$, the running time is $O(f(d)(m+n)\log n)$ with this probability.
\end{proof}

\section{Sd-degeneracy sequences in graphs of low symmetric difference}\label{sec:sd-sd-degen}

We show how to efficiently find, for graphs of moderate symmetric difference, say, $O(n^\gamma)$, an sd-degeneracy sequence of width $O(n^\gamma \log n)$.
The algorithm runs in time $O(n^{3-2\gamma} \log n)$ with high probability.
To balance this running time with the subsequent $O(n^{2+\gamma} \log^2 n)$-time \textsc{All-Pairs Shortest Path} algorithm, we only present it for $\gamma = \frac{1}{3}$.
Note that the next theorem and \cref{cor:sd-degen-to-apsp} imply~\cref{thm:sym-diff}, since $O(n^{\frac{1}{3}} \log n \cdot n^2 \log n)=O(n^{\frac{7}{3}} \log^2 n)$.

\begin{theorem}\label{thm:sym-diff-sd-degen}
  For every constant $c > 0$, there is an algorithm that inputs an $n$-vertex graph $G$ of symmetric difference~$O(n^{\frac{1}{3}})$, and outputs an sd-degeneracy sequence of width $O(n^{\frac{1}{3}} \log n)$, in time $O(n^{\frac{7}{3}} \log n)$ with probability at~least~$1-n^{-c}$.
\end{theorem}

\begin{proof}
  We fix a~positive constant $c > 0$.
  Let $\beta \geqslant 1$ be a~constant such that $s := \beta n^{\frac{1}{3}}$ upper-bounds the symmetric difference of the handled $n$-vertex graphs, and let $G$ be such a~graph.
  To ease some computations, we assume that $n$ is large enough compared to~$c$ and $\beta$.

  The following fact will replace twin-width versatility used in~\cref{thm:vers-tww-to-sd}.
  \begin{claim}\label{clm:many-good-pairs}
    For every induced subgraph~$H$ of~$G$ with at~least $2\lfloor n^{\frac{1}{3}} \rfloor$ vertices, $H$ admits $\lfloor n^{\frac{1}{3}} \rfloor$ pairwise disjoint good pairs.
  \end{claim}
  \begin{claimproof}
  As $H$ has symmetric difference at~most~$s$, it has a~pair $u_1 \neq v_1 \in V(H)$ such that $\sd_H(u_1,v_1) \leqslant s$.
  In turn, the graph $H_2 := H-\{u_1,v_1\}$ admits a~pair $u_2 \neq v_2 \in V(H_2)$ such that $\sd_{H_2}(u_2,v_2) \leqslant s$.
  Thus $\sd_H(u_2,v_2) \leqslant s+2$.
  Then, there is a~pair $u_3 \neq v_3 \in V(H) \setminus \{u_1,v_1,u_2,v_2\}$ such that $\sd_H(u_3,v_3) \leqslant s+4$, and so on.
  This gives $\lfloor n^{\frac{1}{3}} \rfloor$ pairwise disjoint pairs whose symmetric difference is at most $s+2\lfloor n^{\frac{1}{3}} \rfloor$, hence good.
  \end{claimproof}

  \subparagraph*{Algorithm.}
  The algorithm is exactly as~\cref{alg:vers-tww-to-sd} with $g := (\beta+2)n^{\frac{1}{3}} = s+2n^{\frac{1}{3}}$, $\gamma := 2(c+3)(s+2n^{\frac{1}{3}}) \ln n$, and $p := \frac{1}{2}(s+2n^{\frac{1}{3}})^{-1}=O(n^{-\frac{1}{3}})$.
  Again, for every induced subgraph $H$ of~$G$, a~pair $u \neq v \in V(H)$ is called \emph{good} (in $H$) if $\sd_H(u,v) \leqslant g$, and \emph{good-enough} if $\sd_H(u,v) \leqslant \gamma$ (for these new values of $g$ and $\gamma$).
  When $G$ shrinks, $n$ still denotes the initial number of vertices.
  In particular, $s$, $g$, $\gamma$, and $p$ do not change.
  
  \subparagraph*{Correctness and running time.}
  Every pair added to the sd-degeneracy sequence is good-enough in its current induced subgraph of~$G$.
  Thus if the algorithm terminates, it does indeed output an sd-degeneracy sequence of width~$O(n^{\frac{1}{3}} \log n)$.
  We also observe that each iteration is executed on an induced subgraph of $G$, hence a~graph of symmetric difference at~most~$s$.
  With high probability, at every iteration (we will eventually bound the number of iterations by $O(n^{\frac{2}{3}})$), the set $X$ has size at~most~$pn \log n$.
  For later use, we denote by $\mathcal F$ the (unlikely) failure event that at~least~one of the first $64 n^{\frac{2}{3}}$ iterations has a~set $X$ larger than $pn \log n$.
  By Chernoff's bound and union bound, for large enough $n$, we have $\mathbb P(\mathcal F) \leqslant e^{-n^{2/3}}$.

  Conditioned on $\mathcal F$ not holding, we will now show that, excluding the good-enough tests (which cost $O(n^2)$ in total), each iteration takes $O(n^{\frac{5}{3}} \log n)$ time.
  We assume that we have access to the adjacency matrix $M$ of the initial graph $G$.
  As the overall running time is $\Omega(n^2)$, we can do so without loss of generality.
  We represent the induced subgraphs of~$G$ lazily by maintaining a~set of the remaining vertices (and keeping $M$ as is).
  
  The set $X$ takes $O(n)$ time to build.
  The partition $Q$ can be built in time $O(n|X|)=O(n^{\frac{5}{3}} \log n)$: for each remaining vertex $v$, iterate through the vertices of $X$ to build the vector $\vec{v}[X] \in \{0,1\}^X$ of adjacencies w.r.t.~$X$.
  If this vector is measured for the first time (in the current iteration of the while loop), initialize a~new part of~$Q$ as $\{v\}$ at the address $\vec{v}[X]$, and add $\vec{v}[X]$ to the dictionary of existing vectors.
  Otherwise add $v$ to the preexisting part at address $\vec{v}[X]$.

  From the partition representation, forming the pairs can be done in $O(n)$ time.
  Checking that a~pair is good-enough also takes $O(n)$ time.
  We will see that, with high probability, we check $n-O(n^{\frac{1}{3}})$ pairs throughout the entire algorithm.
  Therefore these tests take overall time $O(n^2)$.
  
  We are left with proving that the number of iterations of the while loop is $O(n^{\frac{2}{3}})$ with probability at~least~$1-n^{-c}$.  
  Again, the following holds.
  \begin{enumerate}
  \item\label{it:sd-1} For every good pair $(u, v)$, we have \[\mathbb{P}_X[u\text{~and~}v\text{~are in the same part of~}Q] \geqslant \left(1-p\right)^g = \left(1-\frac{1}{2(s+2n^{\frac{1}{3}})}\right)^{s+2n^{\frac{1}{3}}} \geqslant \frac{1}{2}.\]
  \item\label{it:sd-2} For every $u \neq v$ that does \emph{not} form a~good-enough pair, it holds that \[\mathbb{P}_X[u\text{~and~}v\text{~are in the same part of $Q$}] \leqslant \left(1-p\right)^\gamma \leqslant \exp\left(-\frac{2(c+3)(s+2n^{\frac{1}{3}})\ln n}{2(s+2n^{\frac{1}{3}})}\right)\] \[=e^{-(c+3)\ln n} = n^{-c-3}.\]
  \end{enumerate}

At the beginning of an iteration, the current graph (to ease the notation, we keep denoting it~$G$) has at least $2g > 2 \lfloor n^{\frac{1}{3}} \rfloor$ vertices.
Thus, by~\cref{clm:many-good-pairs}, it admits $h := \lfloor n^{\frac{1}{3}} \rfloor$ disjoint good pairs: $(u_1,v_1), \ldots, (u_h,v_h)$.
Let $Z_i$ be the indicator such that $Z_i=1$ if $u_i$ and $v_i$ are in the same part of~$Q$, and $Z_i=0$ otherwise.
We set $Z := \sum_{i \in [h]} Z_i$.
Thus $\mathbb E[Z] = \sum_{i \in [h]} \mathbb E[Z_i] \geqslant \sum_{i \in [h]} \frac{1}{2} = \frac{h}{2}$, where the inequality follows from~\cref{it:sd-1}.

  We say that an iteration of the while loop is a~\emph{success} if $Z \geqslant \frac{h}{4}$ holds and all the formed pairs (during this iteration) are good-enough.
  In case of a~success, at~least~$\frac{h}{4}$ good pairs have both vertices in the same part of~$Q$.
  In particular, $\sum_{P \in Q} \lfloor \frac{|P|}{2} \rfloor \geqslant \frac{h}{4}$.
  So during such an iteration we form at~least~$\frac{h}{4}$ disjoint pairs.
  This implies that at~least $\frac{h}{4} = \Omega(n^{\frac{1}{3}})$ vertices are deleted at the end of a~success.
  Therefore, there are at~most~$8n^{\frac{2}{3}}$ successes.

  \begin{claim}
    $\mathbb P[\text{the number of iterations~}t\text{~of the while loop is at least~}64 n^{\frac{2}{3}}] \leqslant e^{-2n^{\frac{2}{3}}}$.
  \end{claim}
  \begin{claimproof}
  Let $I_1, \ldots, I_t$ be the indicators of success: for every $i \in [t]$, $I_i=1$ if the $i$th iteration of the while loop is a~success, and $I_i=0$ otherwise.
  By the previous paragraph, it holds that \[\sum_{i \in [t]} I_i \leqslant 8n^{\frac{2}{3}}.\]
  By~\cref{it:sd-2} and union bound over every unordered pair that is not good-enough, the probability that at~least one such pair lands in the same part of~$Q$ is at~most~$\frac{n^2}{2} \cdot n^{-c-3} = \frac{1}{2} n^{-c-1}$.
  Thus by~\cref{clm:Z-tww}, we have $\mathbb P[I_i~|~\mathcal H] \geqslant \frac{1}{3} - \frac{1}{2} n^{-c-1}$ for every $i \in [t]$, where $\mathcal H$ is a~description of what happened in the previous runs.
  Indeed, the lower bound of~\cref{clm:Z-tww} and upper bound of $\frac{1}{2} n^{-c-1}$ are history-independent.
  For large enough $n$, it holds that $\mathbb P[I_i~|~\mathcal H] \geqslant \frac{1}{4}$ for every $i \in [t]$.

  The probability that $t \geqslant 8 \cdot 8n^{\frac{2}{3}}$ is at~most the probability that a~binomial variable $Y \sim B(8t',\frac{1}{4})$ takes value at~most~$t' := \lfloor 8n^{\frac{2}{3}} \rfloor$.
  By Chernoff's bound, the latter is upper bounded by
  \[\mathbb P\left[Y \leqslant \frac{\mathbb E[Y]}{2}\right] \leqslant \exp\left(-\frac{\mathbb E[Y]}{8}\right) \leqslant \exp\left(-2n^{\frac{2}{3}}\right). \claimqedhere \]
  \end{claimproof}
  
  By union bound, the probability that $\mathcal F$ holds or that there are at least $64 n^{\frac{2}{3}}$ iterations of the while loop is at most $\varepsilon := e^{-n^{\frac{2}{3}}} + e^{-2n^{\frac{2}{3}}}$.
  Conditional on that not occurring, the probability that during at~least one iteration, at~least one pair that is not good-enough lands in the same part of the partition $Q$ of that iteration is at~most $64 n^{\frac{2}{3}} \cdot \frac{1}{2} n^{-c-1}$.
  For large enough~$n$, the probability that either of these three bad events holds is at~most \[\varepsilon + (1-\varepsilon) \cdot 64 n^{\frac{2}{3}} \cdot \frac{1}{2} n^{-c-1} \leqslant \frac{1}{2} n^{-c} + \frac{1}{2} n^{-c} = n^{-c}.\]
  Thus, with probability at~least $1-n^{-c}$, the algorithm terminates after fewer than $64 n^{\frac{2}{3}}$ iterations and every formed pair happens to be good-enough.
  (We double-check these pairs for the algorithm to be Las Vegas instead of mere Monte Carlo.)
  In particular, we form and test $n-O(n^{\frac{1}{3}})$ pairs throughout the entire algorithm.
  Furthermore, as $\neg \mathcal F$ holds, each iteration---deprived of the test of pairs which globally costs $O(n^2)$ time---takes $O(n^{\frac{5}{3}} \log n)$ time, hence the overall $64 n^{\frac{2}{3}} \cdot O(n^{\frac{5}{3}} \log n) + O(n^2)= O(n^{\frac{7}{3}} \log n)$ running time.
\end{proof}
  
\section{Matrix multiplication via interval biclique partition}\label{sec:mat-mult}

We denote by $\adj_\prec(G)$ the adjacency matrix of~$G$ along the linear order $\prec$ on~$V(G)$.
We first give a~simple, fast matrix-vector multiplication for $\adj_\prec(G)$ when $G$ comes with an interval biclique partition $(\mathcal B, \prec)$ with few bicliques.

\begin{theorem}\label{thm:ibp-vm-mult}
  There is an algorithm that, given an interval biclique partition $(\mathcal B, \prec)$ of an $n$-vertex graph $G$, and a~column vector $X \in \mathcal A^n$, computes $\adj_\prec(G)X$ over the additive group $\mathcal A$ in time $O(n+|\mathcal B|)$. 
\end{theorem}

\begin{proof}
  We set $M := \adj_\prec(G)$. 
  We start by computing the prefix sums of $X$.
  We set $X_{\leqslant 0} := 0$.
  For every $i$ going from 1 to $n$, we compute $X_{\leqslant i} := X_{\leqslant i-1} + X[i]$.
  This takes $O(n)$ time.
  We then symmetrize $\mathcal B$: for every $(a_1,a_2,b_1,b_2) \in \mathcal B$ (representing the biclique $[a_1,a_2] \times [b_1,b_2]$), we add $(b_1,b_2,a_1,a_2)$ to $\mathcal B$.
  This takes $O(|\mathcal B|)$ time.
  We denote by $\mathcal B'$ the resulting set of bicliques, with $|\mathcal B'| = 2|\mathcal B|$.

  We initialize a~\emph{difference vector} $D := (0, \ldots, 0)$ of length~$n$.
  Eventually, we wish to retrieve $MX$ as the prefix sums of~$D$.
  For every $(a_1,a_2,b_1,b_2) \in \mathcal B'$, we compute $x := X_{\leqslant b_2} - X_{\leqslant b_1 - 1}$, and update $D[a_1] := D[a_1] + x$, and if $a_2 < n$, $D[a_2+1] := D[a_2+1] - x$.
  This loop takes $O(|\mathcal B|)$ time.

  Finally, for $i$ going from 2 to $n$, we do: $D[i] := D[i] + D[i-1]$.
  This takes $O(n)$ time.
  We output $D$.
  The overall running time is $O(n+|\mathcal B|)$.

  \medskip

  We claim that the output $D$ is indeed equal to $MX$.
  Observe that \[(MX)[i] = \sum_{\substack{(a_1,a_2,b_1,b_2)\in\mathcal B'\\ a_1 \le i \le a_2}}\;\sum_{b_1 \le j \le b_2} X[j].\]
  $\sum_{b_1 \le j \le b_2} X[j] = X_{\leqslant b_2} - X_{\leqslant b_1 - 1}$ is indeed the quantity $x$ that we initially add to $D[a_1]$ and subtract from $D[a_2+1]$ (when $a_2 < n$).
  As a~result, when $D$ is output, the contribution of $(a_1,a_2,b_1,b_2)$ is indeed to add $x$ at every index between $a_1$ and $a_2$ (both included).
\end{proof}

We can now wrap up the announced matrix-multiplication algorithm.

\matmult*

\begin{proof}
 Let $G$ be any $n$-vertex graph in~$\mathcal C$, $\prec$ be any linear ordering on~$V(G)$, and $M := \adj_{\prec}(G)$.
 By \cref{thm:vers-tww-to-stm}, we get a~signed tree model of~$G$ with $O(n \log n)$ transversal pairs.
 This takes time $O((|E(G)|+n)\log n)=O(n^2 \log n)$ with high probability.
 By the first item of~\cref{thm:main}, we convert the signed tree model into an interval biclique partition $(\mathcal B,\prec')$ with $|\mathcal B|=O(n \log n)$ in time $O(n \log^2 n)$.

 We compute $M' := \adj_{\prec'}(G)$ in $O(n^2)$ time.
 Let $P$ be the permutation matrix (defined by $\prec, \prec'$) such that $M = PM'P^T$ where $P^T$ is the transpose (thus also the inverse) of~$P$.
 We compute the table of the permutation $\sigma$ (resp.~$\sigma^{-1}$) corresponding to $P$ (resp. to~$P^T$). 
 Let $N_1, N_2, \ldots, N_n$ be the vectors formed by the successive columns of~$N$.
 We will perform each of the $n$ matrix-vector multiplications $MN_i$ in $O(n \log n)$ time.
 
 For each $i \in [n]$, we compute $MN_i$ as $P(M'(P^TN_i))$.
 We use the sparse encoding of $P^T$ (i.e., $\sigma^{-1}$) to compute $N'_i := P^TN_i$ in $O(n)$ time.
 By~\cref{thm:ibp-vm-mult}, we then compute $N''_i := M'N'_i$ in $O(n+|\mathcal B|)=O(n \log n)$ time, using $\mathcal B$.
 We compute $PN''_i$ (equal to $MN_i$) in $O(n)$ time with the sparse encoding $\sigma$ of~$P$.
 We finally arrange the column vectors $MN_i$ in an~$n \times n$ matrix (equal to $MN$).

 The overall running time is $O(n^2 \log n)$.
\end{proof}

\section{Quadratic model checking for merge-width}\label{sec:mw}

In this section, we show that fast algorithms to compute distances can be used to improve the complexity
of the first-order model checking algorithm of Dreier and Toruńczyk which motivates merge-width~\cite{merge-width}.
For a~fixed first-order formula~$\phi$ and a~graph~$G$ on~$n$ vertices given with an appropriate witness of bounded merge-width,
this algorithm tests whether~$G$ satisfies~$\phi$ in time~$O(n^3)$.
At a~very high level, the algorithm follows a~construction sequence,
and computes at each step of this sequence the \emph{local type} of each vertex by dynamic programming.

Deep inside this algorithm, a~subroutine involves computing distances:
greedily finding a~maximal \emph{scattered} sets~$S$,
i.e., a~set of vertices pairwise at distance more than~$r$.
We argue that in the setting this subroutine is used, a~linear-size distance model can be computed using the results of~\cite{bannach2024dag}, resulting in a~speedup from~$O(n^2)$ to~$O(n)$ for this subroutine, and proportionally from~$O(n^3)$ to~$O(n^2)$ for the full algorithm.

While we give pointers to the specific places in~\cite{merge-width} where this scattered set problem appears,
we do so without going into details regarding the definitions and statements using it.
This section thus does not require familiarity with first-order logic and model checking.
\medskip

Call a~construction sequence \emph{positive} if it contains only positive resolves.
Applying \cref{lem:cstr-seq-to-stm} to a~positive merge sequence gives a~positive tree model.
Since a~positive tree model is a~special case of DAG compression, the technique of Bannach, Marwitz, and Tantau~\cite{bannach2024dag} (\cref{thm:DAG-comp}) can directly be applied; the results of \cref{sec:sssp} are not needed here. Thus,
\begin{lemma}\label{lem:pos-cstr-seq-to-dist-model}
    Given a~positive construction sequence for a~graph~$G$ with length $m$ one can compute a~distance model for~$G$ with 0--1 weights of size~$O(m)$, in time~$O(m)$.
\end{lemma}

Computing distances in a~graph given by a~positive construction sequence arises very naturally when considering merge-width.
Indeed, recall that for a~construction sequence $(\Pc_1,E_1,N_1),\dots,(\Pc_m,E_m,N_m)$,
the width is defined by considering at each step~$i$ the distances
in the graph $(V,E_i \cup N_i)$ with all resolved pairs as edges.
\begin{observation}\label{obs:gaifman-pos-cstr-seq}
    For any construction sequence $(\Pc_1,E_1,N_1),\dots,(\Pc_m,E_m,N_m)$ and any step~$i$,
    the graph $(V,E_i \cup N_i)$ is constructed by the \emph{positive} construction sequence
    $(\Pc_1,E_1 \cup N_1, \varnothing), \dots, (\Pc_i,E_i \cup N_i, \varnothing)$.
\end{observation}
The distances in the same graph $(V,E_i \cup N_i)$ are equally important in the model checking algorithm of~\cite{merge-width}, where this graph is the \emph{Gaifman graph} of the relational structure considered.
Specifically, the subroutine they use is the following problem.
Call a~set~$S$ \emph{$r$-scattered} if for any $x,y \in S$, $\dist(x,y) > r$.
\defproblem{\textsc{$c$-Bounded $r$-Scattered Maximal Subset}}{%
    a~graph~$G$ and a~subset~$X$ of vertices.
}{%
    a~subset~$S \subseteq X$ which is (1) $r$-scattered, (2) of size at most~$c$, and (3) is inclusion-wise maximal with these properties.
}
This is an extremely simple problem! It asks to greedily pick vertices in~$X$ pairwise at distance more than~$r$,
and allows to stop as soon as~$c$ (a constant) of them have been found.
The algorithm of~\cite{merge-width} ingeniously uses solutions to this problem
in places where one might expect a~\emph{maximum} $r$-scattered to be needed, rather than an inclusion-wise maximal one.
Specifically, the solutions to this problem define the \emph{scatter types} of \cite[Section~4.2]{merge-width}.

The obvious greedy algorithm for \textsc{$c$-Bounded $r$-Scattered Maximal Subset} is as follows: pick any vertex~$x$ of~$X$ to add to~$S$,
explore the $r$-ball around~$x$ with a~BFS, remove any vertex encountered from~$X$,
and repeat, stopping when either~$X$ is empty or~$c$ vertices have been added.
This takes time $O(cm + n)$ in a~graph with~$n$ vertices and~$m$ edges.
Since~$c$ is a~constant in the merge-width model checking setting, we are not concerned with removing the dependency in~$c$.
However we can use small distance models to improve the running time to~$O(cn)$.

Recall that in the setting of \cref{thm:DAG-comp}, the distance models obtained for a~graph~$G$
are directed graphs with 0--1 edge weights.
In the directed setting, $S$ being $r$-scattered should be understood as:
for $x \neq y$ in~$S$, there is no directed path of length~$r$ from~$x$ to~$y$.
The previous greedy algorithm is easily adapted to this directed setting with 0--1 weights:
the only difference is that one should perform two BFSes from~$x$, one for out-edges and one for in-edges.
Finally, if~$H$ is a~distance model of~$G$, then a~subset $S \subseteq V(G)$ is by definition $r$-scattered in~$G$ if and only if it is $r$-scattered in~$H$.

We thus obtain the following.
\begin{lemma}\label{lem:scattered-set-dist-model}
    If~$G$ is a~graph given by a~distance model with 0--1 weights of size~$m$,
    then \textsc{$c$-Bounded $r$-Scattered Maximal Subset} in~$G$ can be solved in time~$O(cm)$.
\end{lemma}
Let us summarize how this improves the complexity of the algorithm of~\cite{merge-width}, without diving into any of its details.
The main technical result used in their algorithm is the \emph{locality theorem} of \cite[Section~4]{merge-width}.
Its algorithmic statement is \cite[Lemma~4.2]{merge-width}.
While the latter is stated in full generality, in the algorithm of \cite[Section~5]{merge-width},
it is applied to the structures of the form $(\Pc_i,E_i,N_i)$ encountered in the given construction sequence,
with distances being understood in the graph $(V,E_i \cup N_i)$.
By \cref{obs:gaifman-pos-cstr-seq,lem:pos-cstr-seq-to-dist-model}, and the fact that construction sequences always have length linear in the number of vertices (\cref{lem:short-cstr-seq}),
one can thus compute in time~$O(n)$ distance models of size~$O(n)$ for these structures.

Now the complexity of \cite[Lemma~4.2]{merge-width} comes entirely from a~bounded number of calls to \cite[Observation~4.8]{merge-width}, which deals with computing scatter types.
As previously mentioned, these scatter types are themselves defined by the solutions to a~bounded number of calls to 
\textsc{$c$-Bounded $r$-Scattered Maximal Subset}, with constant~$c$.
Thus, assuming an~$O(n)$ size distance model to be given, the scatter type is computed in~$O(n)$ time by \cref{lem:scattered-set-dist-model}.
Thus, \cite[Lemma~4.2]{merge-width} can be implemented in time~$O(n)$ when the distance model is given,
improving upon the running time claimed in~\cite{merge-width} which was linear in the number of edges.

In practice, \cite[Lemma~4.2]{merge-width} is applied in the algorithm of \cite[Section~5]{merge-width}
to graphs with a~quadratic number of edges.
There, we thus improve the running time from~$O(n^2)$ to~$O(n)$.
One may check that this is the only bottleneck, and the merge-width algorithm \cite[Theorem~1.11]{merge-width}
thus improves proportionally from~$O(n^3)$ to~$O(n^2)$.

\end{document}